\documentclass{article}%
\usepackage{amssymb}
\usepackage{amsfonts}
\usepackage{amsmath}
\usepackage{graphicx}%
\providecommand{\U}[1]{\protect\rule{.1in}{.1in}}
\newtheorem{proposition}{Proposition}[section]
\newtheorem{theorem}[proposition]{Theorem}

\newtheorem{lemma}[proposition]{Lemma}

\newenvironment{proof}[1][Proof]{\noindent\textbf{#1.} }{\hfill$\Box$\vspace{0.4cm}}
\providecommand{\U}[1]{\protect\rule{.1in}{.1in}}
\DeclareMathAlphabet{\mathpzc}{OT1}{pzc}{m}{it}

\newcommand{\eqnsection}{\renewcommand{\theequation}{\thesection.\arabic{equation}}
\makeatletter \csname @addtoreset\endcsname{equation}{section}\makeatother}

\def\l|{\left|\left|}
\def\r|{\right|\right|}

\def\1{\mathds 1}

\begin{document}
\begin{titlepage}
\title{Reconstruction and Clustering\\in Random Constraint Satisfaction Problems}
\author{Andrea Montanari\thanks{~Departments of Electrical Engineering and Statistics,
Stanford University; research funded in part by the NSF grants
CCF-0743978 and DMS-0806211}
\and Ricardo Restrepo\thanks{~School of Mathematics, Georgia Tech, Atlanta, GA
30332-0160}
\and Prasad Tetali\thanks{~Schools of Mathematics and Computer Science, Georgia
Tech, Atlanta, GA 30332-0160; research funded in part by the NSF grant
DMS-0701043}$^{\dagger}$}
\date{\today}
\maketitle
\begin{abstract}
Random instances of Constraint Satisfaction Problems (CSP's)
appear to be hard for all known algorithms, when the number
of constraints per variable lies in a certain interval.
Contributing to the general understanding of the structure
of the solution space of a CSP in the satisfiable regime, we formulate a set
of natural technical conditions on a large family of (random) CSP's, and prove
bounds on three most interesting thresholds for the density of such an
ensemble: namely, the \emph{satisfiability} threshold, the threshold for
\emph{clustering} of the solution space, and the threshold for an appropriate
\emph{reconstruction} problem on the CSP's. The bounds become asymptoticlally
tight as the number of degrees of freedom in each clause diverges.
The families are general enough to include commonly studied problems such as,
random instances of Not-All-Equal-SAT, $k$-XOR formulae, hypergraph 2-coloring, and graph $k$-coloring. An important new ingredient is a condition involving the Fourier
expansion of clauses, which characterizes the class of problems with a similar
threshold structure.
\end{abstract}
\end{titlepage}

\section{Introduction}

Given a set of $n$ variables taking values in a finite alphabet, and a
collection of $m$ constraints, each restricting a subset of variables, a
Constraint Satisfaction Problem (CSP) requires finding an assignment to the
variables that satisfies the given constraints. Important examples include
$k$-SAT, Not All Equal SAT, graph (vertex) coloring with $k$ colors etc.
Understanding the threshold of satisfiability/unsatisfiability for
\emph{random} instances of CSPs, as the number of constraints $m=m(n)$ varies,
has been a challenging task for the past couple of decades, with some notable
successes (see e.g., \cite{ANP05}). On the algorithmic side, the challenge of
\emph{finding} solutions of a random CSP \emph{close to the threshold of
satisfiability} (in the regime where solutions are known to exist) remains
widely open. All provably polynomial-time algorithms fail well before the SAT
to UNSAT threshold.

The attempt to understand this universal failure led to studying the geometry
of the set of solutions of random CSPs \cite{MarcGiorgioRiccardo,AchlioCojas},
as well as the emergence of long range correlations among variables in random
satisfying assignments \cite{OurPNAS}. These research directions are motivated
by two heuristic explanations of the failure of polynomial algorithms:
\textbf{(1)} The space of solutions becomes increasingly complicated as the
number of constraints increases and is not captured correctly by simple
algorithms; \textbf{(2)} Typical solutions become increasingly correlated and
local algorithms cannot unveil such correlations.

By analyzing a large class of random CSP ensembles, this paper provides strong
support to the belief that the above phenomena are \emph{generic}, that they
are characterized by \emph{sharp thresholds}, and that the 
\emph{thresholds for clustering and reconstruction do coincide}.

\subsection{Related work}

Building on a fascinating conjecture on the geometry of the set of solutions,
statistical physicists have developed surprisingly efficient message passing
algorithms to solve random CSPs. For instance, survey propagation
\cite{MarcGiorgioRiccardo,MarcRiccardo} has been shown empirically to find
solutions of random 3-SAT extremely close to the SAT-UNSAT transition. In
order to understand the success of these heuristics, it has become important
to study the thresholds for the emergence of so-called \emph{clustering} of
solutions -- the emergence of an exponential number of sets (or clusters) of
solutions, where solutions within a cluster are closer (in the Hamming
sense, say), compared to the intra-cluster distance
\cite{Mora,AchlioRicci,AchlioCojas}. Moreover, the fact that solutions within
a cluster impose long-range correlations among assignments of variables,
motivates one to study the so-called reconstruction problem in the context of
random CSP's. Indeed, non-rigorous statistical mechanics calculations imply
that the clustering and reconstruction thresholds coincide
\cite{MezardMontanari,OurPNAS}.

Finally, understanding the threshold for (non)reconstruction is also becoming
relevant (if not crucial) to understanding the limit of the Glauber dynamics
to sample from the set of solutions of a CSP. Indeed non-reconstuctibility was
proved in \cite{PeresEtAl} to be a necessary condition for fast mixing, and is
expected to be sufficient for a large class of `sufficiently random' problems
\cite{mon-ger}.

In a recent paper, Gerschenfeld and the first author \cite{mon-ger},
considered the reconstruction problem for \emph{graphical models}, which
included the case of proper colorings of the vertices of a random graph. This
amounts to understanding the correlation (as measured e.g. through mutual
information) between the color of a vertex $v$, and the colors of vertices at
distance $\geq t$ from $v$. In particular, the problem is said to be
`unsolvable' if such a correlation decays to $0$ with $t$. We refer to
Section~\ref{sec:MainResults} for a precise definition of the reconstruction
problem. For a class of models, including the so-called Ising spin glass, the
antiferromagnetic Potts model, and proper $q$-colorings of a graph,
\cite{mon-ger} derived a general sufficient condition, under which
reconstruction for (sparse) random graphs $G(n,m)$ with $m=cn$ edges is
possible if and only if it is possible for a Galton-Watson tree with
independent $\operatorname*{Poisson}(2c)$ degrees for each vertex. Moreover,
they also verified that the condition holds for the Ising spin glass and the
antiferromagnetic Potts at non-zero temperature, leaving open the case of
proper colorings of graphs.

\subsection{Summary of contributions}

It is against this backdrop that we consider certain general families of CSP's
-- the first dealing with constraints consisting of $k$-tuples of binary
variables (as in $k$-uniform hypergraph 2-coloring or Not-All-Equal (NAE)
$k$-sat), while the second dealing with $q$-colorings of vertices of graphs
(which may be seen as an instance of a CSP with $q$-ary variables) -- and
study three important threshold phenomena. Our chief contribution is as follows.

\textbf{(a)} We formulate a fairly natural set of assumptions under which a
general class of constraint satisfaction problems (including the models
mentioned above) can be understood rather precisely in terms of the thresholds
for satisfiability, clustering and (non)reconstruction phenomena. In
particular we verify that the last two thresholds coincide within the
precision of our bounds.

\textbf{(b)} We consider tree ensembles (families of random CSP's whose
variable-constraint dependency structure takes the form of a tree), and prove
optimal bounds on the threshold for reconstruction on trees. These CSP's
consist of binary variables, and the constraints are $k$-ary, and the bounds
are optimal to first order, as $k$ goes to infinity.

\textbf{(c)} We verify the sufficient condition of \cite{mon-ger} for proper
colorings of graphs, thus extending the reconstruction result for colorings on
trees to the same on (sparse) random graphs.

\textbf{(d)} By way of techniques, we make crucial use of the Fourier
expansion of the (binary $k$-CSP) constraints, after introducing an assumption
on the Fourier expansion, as part of the random ensemble under consideration;
this is key to being able to characterize the thresholds precisely.

\textbf{(e)} Finally, as illustrative examples, we mention the specific bounds
(on various thresholds) that follow for some standard models, such as the NAE
$k$-SAT, $k$-XOR formulae etc.

\ 

The organization of the paper is as follows. In Section~\ref{sec:FormalDefns},
we give the formal definitions and assumptions of our models. We state our
main results in Section \ref{sec:MainResults}. In Section~\ref{sec:TreeRec},
we state and prove the optimal bounds for the tree reconstruction problem. In
Section~\ref{sec:ColorRec}, we verify the sufficient condition (from
\cite{mon-ger}) for the specific problem of graph proper $q$-coloring, thus
proving one of our main results -- optimal bounds on the (sparse) random graph
reconstruction problem for colorings. In Appendix~\ref{sec:SecondMom}, we
derive a certain technical second moment bound that is needed for our work.

\section{Definitions}

\label{sec:FormalDefns}

In this section we define a family of random CSP ensembles: problems with
constraints involving $k$-tuples of binary variables and $q$-ary ensembles as
a natural extension. We also introduce some analytic definitions that we will
need in order to present our results.

\bigskip

\emph{Binary }$k$\emph{-CSP ensemble.} Given an integer $n$, $\alpha
\in\mathbb{R}_{+}$, and a distribution $p=\{p(\varphi)\}$ over Boolean
functions $\varphi:\{+1,-1\}^{k}\rightarrow\{0,1\}$, $\operatorname*{CSP}%
(n,\alpha,p)$ is the ensemble of random CSP's over $n$ Boolean variables
$\underline{x}=(x_{1},\dots,x_{n})$ defined as follows. For each
$a\in\{1,\dots,m=n\alpha\}$, draw $k$ indices $i_{a}(1),\dots,i_{a}(k)$
independently and uniformly at random in $[n]$, and a function $\varphi_{a}$
with distribution $p(\varphi)$. An assignment $\underline{x}$ satisfies the
resulting instance if $\varphi_{a}(x_{i_{a}(1)},\dots,x_{i_{a}(k)})=1$ for
each $a\in\lbrack m]$. A CSP instance can be naturally described by a
bipartite graph $G$ (often referred to in the literature as a `factor graph')
including a node for each clause $a\in\lbrack m]$ and for each variable
$i\in\lbrack n]$, and an edge $(i,a)$ whenever variable $x_{i}$ appears in the
$a$-th clause.

\bigskip

$q$\emph{-ary ensembles. } A $q$-ary ensemble is the natural generalization of
a binary ensemble to the case in which variables take $q$ values. For the sake
of simplicity, we restrict our discussion here to the case of pairwise
constraints (i.e. $k=2$ in the language of the previous section).

Given an integer $n$, $\alpha\in\mathbb{R}_{+}$, and a distribution
$p=\{p(\varphi)\}$ over Boolean functions $\varphi:[q]\times\lbrack
q]\rightarrow\{0,1\}$, $\operatorname*{CSP}_{q}(n,\alpha,p)$ is the collection
of random CSP's over $q$-ary variables $x_{i}$, for $i=1,2,\ldots,n$, defined
as follows. For each $a\in\{1,\dots,m=n\alpha\}$, draw $2$ indices
$i_{a},j_{a}$ independently and uniformly at random in $[n]$, and a function
$\varphi_{a}$ with distribution $p(\varphi)$. An assignment $\underline
{x}=(x_{1},\dots,x_{n})$ satisfies the resulting instance, if $\varphi
_{a}(x_{i_{a}},x_{j_{a}})=1$ for each $a\in\lbrack m]$.

In this paper, by way of illustrating how the results for binary ensembles
could be (purportedly) extended to $q$-ary ensembles, we will exclusively
study the $q$-coloring model which consists of ensembles with the single
clause $\varphi\left(  x,y\right)  =\mathbb{I}\left(  x\neq y\right)  $. This
model corresponds to proper colorings with $q$ colors of a random sparse graph
with an edge-to-vertex density of $\alpha>0$.

\bigskip

In the rest of this section, we briefly review some well known definitions in
discrete Fourier analysis that are useful for stating our results.

\ 

\emph{Functional analysis of clauses.} We denote by $v_{\theta}$, the measure
defined over $\{-1,+1\}^{k}$ such that $v_{\theta}\left(  x\right)  =%
{\textstyle\prod\limits_{i=1}^{k}}
\left(  \frac{1+x_{i}\theta}{2}\right)  $ for every $x\in\{-1,+1\}^{k}$. This
is just the measure induced by choosing $k$ independent copies of a random
variable that takes values $\pm1$ and has expectation $\theta$. Notice that
when $\theta=0$, $v_{\theta}$ corresponds to the uniform measure over
$\{-1,+1\}^{k}$.

The inner product induced by this measure, on the space of real functions
defined on $\{-1,+1\}^{k}$ is denoted by $\left(  \cdot,\cdot\right)
_{\theta}$, and the correponding norm by $\left\Vert \cdot\right\Vert
_{\theta}$. If $\theta=0$, we drop the subindex and just use $\left(
\cdot,\cdot\right)  $ and $\left\Vert \cdot\right\Vert $, respectively. Thus,
if $f,g:\{-1,+1\}^{k}\rightarrow\mathbb{R}$, then%
\begin{align*}
\left(  f,g\right)  _{\theta}  &  =\sum_{x\in\{-1,+1\}^{k}}f\left(  x\right)
g\left(  x\right)  v_{\theta}\left(  x\right)  \text{,\quad}\left\Vert
f\right\Vert _{\theta}^{2}=\sum_{x\in\{-1,+1\}^{k}}f^{2}\left(  x\right)
v_{\theta}\left(  x\right)  \text{,}\\
\left(  f,g\right)   &  =2^{-k}\sum_{x\in\{-1,+1\}^{k}}f\left(  x\right)
g\left(  x\right)  \text{,\quad}\left\Vert f\right\Vert ^{2}=2^{-k}\sum
_{x\in\{-1,+1\}^{k}}f^{2}\left(  x\right)  \text{.}%
\end{align*}
We denote the Hilbert space of functions $\{-1,+1\}^{k}\rightarrow\mathbb{R}$
under the inner product $\left(  \cdot,\cdot\right)  $ by $J_{k}$.

\bigskip

\emph{Fourier transform of clauses.} For any $Q\subseteq[k]\equiv\left\{
1,\ldots,k\right\}  $, let $\gamma_{Q}(x)\overset{def}{=}\prod_{i\in Q}x_{i}$.
Under the scalar product defined above (with $\theta=0$), the functions
$\left\{  \gamma_{S}\right\}  _{S\subseteq[k]}$ form an orthonormal basis for
$J_{k}$. Moreover, they are exactly the algebraic characters of $\left\{
-1,1\right\}  ^{k}$ with the group operation of pointwise multiplication.
Thus, we define the Fourier transform of a function $f\in J_{k}$, by letting
for any $Q\subseteq\left[  k\right]  $,
\[
f_{Q}\overset{def}{=}(\gamma_{Q},f)=2^{-k}\sum_{x\in\{-1,+1\}^{k}}%
f(x)\gamma_{Q}(x)\text{.}%
\]

\bigskip

\emph{Noise operator. }Given $\theta\in\left[  -1,1\right]  $, we define the
\emph{Bonami - Beckner} operator $\operatorname*{T}_{\theta}:J_{k}\rightarrow
J_{k}$, by
\[
\left(  \operatorname*{T}\nolimits_{\theta}f\right)  \left(  x\right)
\overset{def}{=}\sum_{y\in\left\{  -1,1\right\}  ^{k}}f\left(  xy\right)
v_{\theta}\left(  y\right)  \text{.}%
\]
Notice that $\left(  \operatorname*{T}_{\theta}f\right)  (x)$ corresponds to
the expected value of $f(\mathbf{x}_{\theta})$, where $\mathbf{x}_{\theta}$ is
obtained from $x$ by flipping each coordinate independently with probability
$(1-\theta)/2$. Notice that $\operatorname*{T}_{1}$ is just the identity
operator and $\operatorname*{T}_{0}$ sends $f$ to the constant function
$\left(  f,\gamma_{\emptyset}\right)  $.

The Bonami-Beckner operator diagonalizes with respect to the Fourier basis, in
the sense that $\left(  \operatorname*{T}\nolimits_{\theta}\gamma_{Q}\right)
\left(  x\right)  =\theta^{\left\vert Q\right\vert }\gamma_{Q}\left(
x\right)  $ for any $Q\subseteq\left[  k\right]  $.

More generally, given $h\in\left[  -1,1\right]  ^{k}$, we define $\left(
\operatorname*{T}_{h}f\right)  (x)\overset{def}{=}\mathbb{E}[f(\mathbf{x}%
_{h})]$, where $\mathbf{x}_{h}$ is obtained from $x$ by flipping the $i^{th}$
coordinate independently and with probability $\frac{1-h_{i}}{2}$. Since
$\operatorname*{T}_{h}$ also diagonalizes with respect to the Fourier basis,
one gets $\left(  \operatorname*{T}\nolimits_{h}\gamma_{S}\right)  \left(
x\right)  =\gamma_{S}\left(  h\right)  \gamma_{S}\left(  x\right)  .$

\bigskip

\emph{Discrete derivative and influence. }Given a function $f\in J_{k-1}$, we
define its \emph{discrete derivative} $f^{\left(  1\right)  }\in J_{k-1}$ as
$f^{\left(  1\right)  }\left(  x\right)  =\frac{1}{2}\left[  f\left(
1,x\right)  -f\left(  -1,x\right)  \right]  $. We define analogously
$f^{\left(  i\right)  }$ for any other variable index. Finally, the
\emph{influence} of the $i^{\text{th}}$ variable on $f$ is defined using the
norm of the derivative
\[
\operatorname*{I}\nolimits_{i}\left(  f\right)  \overset{def}{=}\left\Vert
f^{\left(  i\right)  }\right\Vert ^{2}\text{.}%
\]
For any $Q\subseteq\left[  k\right]  $, $f_{Q}^{\left(  i\right)  }%
=f_{Q\cup\left\{  i\right\}  }$.

\section{Main results}

\label{sec:MainResults}

\subsection{Binary $k$-CSP ensembles}

\label{sec:BinaryCSP}

We assume the following conditions on the ensemble.

\bigskip

\emph{1. Permutation symmetry.} If $\varphi^{\pi}$ is the Boolean function
obtained from $\varphi$ by permuting its arguments, we require $p(\varphi
^{\pi})=p(\varphi)$.

\bigskip

\emph{2. Balance.} The distribution $p$ is supported on Boolean functions such
that $\varphi(x_{1},\dots,x_{k})=\varphi(-x_{1},\dots,-x_{k})$. This condition
implies that the odd Fourier coefficients of $\varphi$ are zero.

\bigskip

\emph{3. Feasibility.} For each Boolean function $\varphi$ in the support of
$p$, every partial assignment $\left(  x_{1},\ldots,x_{k-1}\right)  $ can be
extended to a satisfying assignment $\left(  x_{0},x_{1},\ldots,x_{k-1}%
\right)  $ of $\varphi$. This condition implies that $\left\Vert
\varphi\right\Vert ^{2}\geq1/2$, and together with the balance condition,
implies that all the variables of $\varphi$ have the same influence, namely,
$\operatorname*{I}_{i}\left(  \varphi\right)  =\frac{1-||\varphi||^{2}}{2}$.

\bigskip

\emph{4. Dominance of balanced assignments.} For every $\theta\in\left[
-1,1\right]  $,
\[
\mathbb{E}_{\varphi}\log\left\Vert \varphi\right\Vert _{\theta}\leq
\mathbb{E}_{\varphi}\log\left\Vert \varphi\right\Vert \text{.}%
\]
This condition implies that, in a typical random instance, most solutions are
balanced in the sense that they have almost as many $+1$'s as $-1$%
's.\vspace{0.2cm}

\bigskip

While our ultimate goal is to exhibit results as $k\rightarrow\infty$, the
probability distribution $p$ over the functions $\varphi:\left\{
-1,1\right\}  ^{k}\rightarrow\left\{  0,1\right\}  $ must be defined for
\emph{every} $k$, and some agreement should exist between such probability
distributions for different $k$'s. In our work this agreement is given by two
conditions concerning the derivative of the clauses in the support of $p$:

$\bigskip$

\noindent\textbf{(a)} $l_{1}$ \emph{norm of the Fourier transform grows at
most polynomially in }$k$. That is, for every $\varphi\in\operatorname*{supp}%
(p)$,
\begin{equation}
\sum_{Q}\left\vert \varphi_{Q}^{\left(  i\right)  }\right\vert \leq k^{a}\;,
\label{fourier2}%
\end{equation}
for some constant $a$ not depending on $k$.


\bigskip

\noindent\textbf{(b)} \emph{`Small weight'} \emph{Fourier coefficients are
small.} There is a constant $C>0$ (not depending on $k$) such that for every
$\varphi\in\operatorname*{supp}\left(  p\right)  $,
\begin{equation}
\;\;\;\;\left\Vert \operatorname*{T}\nolimits_{\theta}\varphi^{\left(
i\right)  }\right\Vert ^{2}\leq e^{-Ck\,\left(  1-\theta\right)  }\left\Vert
\varphi^{\left(  i\right)  }\right\Vert ^{2}\,\text{, }\theta\in\left[
0,1\right]  \text{.}\label{fourier}%
\end{equation}
The above implies in particular, that for any fixed $\ell$, there exists
$A_{\ell}>0$ (independent of $k$), such that
\begin{equation}
\sum_{1\leq|Q|\leq\ell}|\varphi_{Q}|^{2}\leq A_{\ell}e^{-Ck/2}\sum_{\left\vert
Q\right\vert \geq1}|\varphi_{Q}|^{2}\text{.}\label{fourier3}%
\end{equation}
An equivalent formulation of Eq. (\ref{fourier}) (with a possibly different
constant $C$) is
\begin{equation}
\;\;\;\;\left(  \operatorname*{T}\nolimits_{\theta}\varphi^{\left(  i\right)
},\varphi^{\left(  i\right)  }\right)  \leq e^{-Ck\,\left(  1-\theta\right)
}\left\Vert \varphi^{\left(  i\right)  }\right\Vert ^{2}\,\text{, }\theta
\in\left[  0,1\right]  \text{.}\label{fourier4}%
\end{equation}


\bigskip

\noindent\textbf{Results.} An ensemble of binary $k$-CSP's will be
characterized by the following quantities.
\[
\frac{1}{\Omega_{k}}\overset{def}{=}\mathbb{E}_{\varphi}\frac
{2\operatorname*{I}_{1}\left(  \varphi\right)  }{\left\Vert \varphi\right\Vert
^{2}}\,\text{,\qquad}\frac{1}{\widehat{\Omega}_{k}}\overset{def}{=}%
-\mathbb{E}_{\varphi}\log\Bigl(  \left\Vert \varphi\right\Vert ^{2}\Bigr)
\,.
\]
Notice that $\Omega_{k}\leq\widehat{\Omega}_{k}$ and $\Omega_{k}%
\approx\widehat{\Omega}_{k}$, whenever the influence is relatively small, or
equivalently, when the norm is close to $1$.

\begin{proposition}
\label{propo:SATCSP} A random binary constraint satisfaction instance from the
$\operatorname*{CSP}(n,\alpha,p)$ ensemble is satisfiable, with high
probability, if $\alpha<\alpha_{\mathrm{s}}(k)$, where
\[
\Omega_{k}\,\log2\;\;\{1+o(1)\}\leq\alpha_{\mathrm{s}}(k)\leq\widehat{\Omega
}_{k}\,\log2\;\;\{1+o(1)\}\,.
\]
Vice versa, if $\alpha>\alpha_{\mathrm{s}}(k)(1+o(1))$, then with high
probability, a $\operatorname*{CSP}(n,\alpha,p)$ instance is unsatisfiable.
\end{proposition}

Given an instance of $\operatorname*{CSP}(n,\alpha,p)$, a cluster of solutions
is any equivalence class of solutions under the (closure of the) relation
$\underline{x}\simeq\underline{x}^{\prime}$ if $d_{\mathrm{Hamming}%
}(\underline{x},\underline{x}^{\prime})\leq d_{\mathrm{max}}$ for some
$d_{\mathrm{max}}=o(n)$. The set of solutions is \emph{clustered} if it is
partitioned into exponentially many clusters.

\begin{theorem}
\label{thm:ClustCSP} The set of solutions of an instance from the
$\operatorname*{CSP}(n,\alpha,p)$ ensemble is clustered, with high
probability, if $\alpha>\alpha_{\mathrm{d}}(k)$, where
\[
\alpha_{\mathrm{d}}(k)=\frac{\Omega_{k}}{k}\,\{\log k+\operatorname*{o}(\log
k)\}\,.
\]

\end{theorem}

Given a measure $\mu(\underline{x})$ over variable assignments in
$\{+1,-1\}^{V}$, the reconstruction problem is said to be unsolvable if
correlations with respect to $\mu$ decay rapidly with the distance $r$ on $G
$. More precisely, if $\mu_{i,\sim r}$ denotes the joint distribution of
$x_{i}$ and $\{x_{j}:\,d_{G}(i,j)\geq r\}$, then $\ \lim_{r\rightarrow\infty
}\lim\sup_{n\rightarrow\infty}\mathbf{E}\Vert\mu_{i,\sim r}-\mu_{i}\mu_{\sim
r}\Vert_{\operatorname*{TV}}=0$.

\begin{theorem}
\label{reconstruction}Let $\mu(\underline{x})$ be the uniform measure over
solutions of an instance from the $\operatorname*{CSP}(n,\alpha,p)$ ensemble.
The reconstruction problem is solvable
for $\mu$ if $\alpha>\alpha_{\mathrm{r}}(k)$, where
\[
\alpha_{\mathrm{r}}(k)=\frac{\Omega_{k}}{k}\,\{\log k+\operatorname*{o}(\log
k)\}\,.
\]
Vice versa, the reconstruction problem is unsolvable if $\alpha<\alpha
_{\mathrm{r}}(k)$.
\end{theorem}

Thus, a key result of the present paper is that $\alpha_{\mathrm{d}}(k)$ and
$\alpha_{\mathrm{r}}(k)$ do \emph{coincide for a large family of ensembles}
(up to the slackness, in the second order terms, of our bounds).

\bigskip

\textbf{Example:\ 2-coloring hypergraphs}. Let us consider the ensemble of
CSP's consisting of clauses of the type $\varphi$, where $\varphi\left(
x_{1},\ldots,x_{k}\right)  =\mathbb{I}\left(
{\textstyle\sum}
x_{i}\notin\left\{  -k,k\right\}  \right)  $. The $\operatorname*{CSP}%
(n,\alpha,p)$ in this case, corresponds to the distribution of 2-colorings of
a random hypergraph on $n$ vertices and $\alpha n$ edges, with edge size $k$,
and each edge chosen independently and uniformly at random.

The conditions 1-3 clearly hold for this model and the dominance of balance
assignments follows after checking that $\left\Vert \varphi\right\Vert
_{\theta}=1-\left(  \frac{1+\theta}{2}\right)  ^{k}-\left(  \frac{1-\theta}%
{2}\right)  ^{k}$ maximizes at $\theta=0$. To establish the conditions
(\ref{fourier2}), notice that $\varphi_{Q}^{\left(  i\right)  }=-\frac
{1}{2^{k}}[1-\left(  -1\right)  ^{\left\vert Q\right\vert }]$, which clearly
implies that the $l_{1}$ norm of the fourier transform is bounded. To check
(\ref{fourier}), notice that $\;\frac{\left(  \operatorname*{T}%
\nolimits_{\theta}\varphi^{\left(  i\right)  },\varphi^{\left(  i\right)
}\right)  }{\left\Vert \varphi^{\left(  i\right)  }\right\Vert ^{2}\,}=\left(
\frac{1+\theta}{2}\right)  ^{k-1}-\left(  \frac{1-\theta}{2}\right)
^{k-1}\leq e^{-k\left(  1-\theta\right)  /2}$ for all $\theta\in\left[
0,1\right]  $. 

\ 

An easy computation shows that $\Omega_{k}=2^{k-1}-1$ and $\frac{1}%
{\widehat{\Omega}_{k}}=-\log(1-2^{-k+1})$, therefore we have:%

\

\begin{tabular}
[c]{|c|c|c|c|}\hline
& $\text{Reconstruction - Clustering}$ & $\text{Lower bound satisfiability}$ &
$\text{Upper bound satisfiability}$\\\hline
$\text{2-coloring}$ & $\frac{2^{k-1}}{k}\,\left[  \log k+\operatorname*{o}%
(\log k)\right]  $ & $2^{k-1}\log2\left[  1+\operatorname*{o}(1)\right]  \,$ &
$2^{k-1}\log2\left[  1+\operatorname*{o}(1)\right]  $\\\hline
\end{tabular}

\bigskip

\textbf{Example: Not All Equal $k-$SAT}. Let us consider now an ensemble of
CSP's consisting of clauses of type $\left\{  \varphi_{s}\right\}
_{s\in\left\{  +1,-1\right\}  ^{k}}$, where $\varphi_{s}\left(  x_{1}%
,\ldots,x_{k}\right)  =\mathbb{I}\left(
{\textstyle\sum}
x_{i}s_{i}\notin\left\{  -k,k\right\}  \right)  $ and $p\left(  \varphi
_{s}\right)  =2^{-k}$ for each $s\in\left\{  +1,-1\right\}  ^{k}$. In this
case, the $\operatorname*{CSP}\left(  n,\alpha,p\right)  $ model corresponds
to the distribution of NAE $k-$SAT instances for a random formula in $n$
variables, consisting of $\alpha n$ random clauses, each with $k$ literals.

For this model, the conditions 1-3 are easily verified. The dominance of
balance assignments follows from
\[
\mathbb{E}_{s}\log\left\Vert \varphi\right\Vert _{\theta}\leq\log
\mathbb{E}_{s}\left\Vert \varphi\right\Vert _{\theta}=\log\mathbb{E}%
_{s}\left(  1-%
{\textstyle\prod\nolimits_{i=1}^{k}}
\frac{1+s_{i}\theta}{2}-%
{\textstyle\prod\nolimits_{i=1}^{k}}
\frac{1-s_{i}\theta}{2}\right)  =\mathbb{E}_{s}\log\left\Vert \varphi
\right\Vert \text{.}%
\]
On the other hand, the Fourier expansion of $\varphi_{s}$ is given by
$\varphi_{s,Q}=-2^{-k}[\gamma_{Q}(s)+\gamma_{Q}(-s)]$. In particular
$\left\vert \varphi_{s,Q}\right\vert ^{2}=2^{-k}[1+\left(  -1\right)
^{\left\vert Q\right\vert }]$, so that both Eqs.~(\ref{fourier2}) and
(\ref{fourier}) hold along the same lines as the previous example. Indeed, in
this case we get the same values for $\Omega_{k}$ and $\widehat{\Omega}_{k}$,
so that, we have:

\

\begin{tabular}
[c]{|c|c|c|c|}\hline
& $\text{Reconstruction - Clustering}$ & $\text{Lower bound satisfiability}$ &
$\text{Upper bound satisfiability}$\\\hline
$\text{NAE-SAT}$ & $\frac{2^{k-1}}{k}\,\left[  \log k+\operatorname*{o}(\log
k)\right]  $ & $2^{k-1}\log2\left[  1+\operatorname*{o}(1)\right]  $ &
$2^{k-1}\log2\left[  1+\operatorname*{o}(1)\right]  $\\\hline
\end{tabular}

\bigskip

\textbf{Example:\ }$k$\textbf{-XOR formulas. } For an even integer $k$, the
$k$-XOR ensemble ($k$ even) consists of clauses of type $\left\{
\varphi_{\epsilon}\right\}  _{\epsilon=1,-1}$, where $\varphi_{\epsilon
}\left(  x_{1},\ldots,x_{k}\right)  =\frac{1}{2}\left(  \gamma_{\emptyset
}+\epsilon\gamma_{\left[  k\right]  }\right)  $. In this case, the
$\operatorname*{CSP}\left(  n,\alpha,p\right)  $ model corresponds to a system
of $\alpha n$ random linear equations in $\mathbb{Z}_{2}$, in which every
equation involves $k$ randomly chosen variables (with replacement) from a
total of $n$ possible variables.

Conditions 1-3 hold for $k$ even, and the dominance of balanced assignments
condition follows from the fact that $\mathbb{E}_{\varphi}\log\left\Vert
\varphi\right\Vert _{\theta}=\frac{1}{2}\log\left(  \frac{1-\theta^{2k}}%
{4}\right)  $, which is clearly maximized at $\theta=0$. The condition on
Fourier expansion of clauses for this model is straightforward:\ The Fourier
expansion of $\varphi_{\epsilon}$ is concentrated at $\emptyset$ and $\left[
k\right]  $, so that the Eq. (\ref{fourier2}) holds with $a=0$ and the Eq.
(\ref{fourier2}) holds with $C=1$.

In this case, we have that $\Omega_{k}=1$, while $\widehat{\Omega}_{k}%
=1/\log2$. Therefore, we have: \bigskip%

\begin{tabular}
[c]{|c|c|c|c|}\hline
& $\text{Reconstruction - Clustering}$ & $\text{Lower bound satisfiability}$ &
$\text{Upper bound satisfiability}$\\\hline
$\text{XOR-SAT}$ & $\frac{1}{k}\,\left[  \log k+\operatorname*{o}(\log
k)\right]  $ & $\log2+\operatorname*{o}(1)$ & $1+\operatorname*{o}(1)$\\\hline
\end{tabular}

\

We remark here that, in the case of XOR-SAT, the clustering and satisfiability
thresholds can be determined \emph{exactly} by exploiting  the underlying
group structure \cite{XOR1,XOR2} (see \cite{MezardMontanariBook} for a
discussion of the reconstruction problem  in XOR-SAT).

\subsection{$q$-ary ensembles: graph coloring}

The following result concerning the colorability and clustering of proper
colorings were proved by Achlioptas and Naor \cite{achlioassaf} and Achlioptas
and Coja-Oghlan \cite{AchlioCojas}.
%

\begin{theorem}
\label{thm:qcol} \textrm{(Graph $q$-colorability \cite{achlioassaf}) } A
random graph with $n$ vertices and $n\alpha$ edges is satisfiable with high
probability if $\alpha<\alpha_{\mathrm{s}}(q)$, where
\[
\alpha_{\mathrm{s}}(q)=q\left[  \log q+o_{q}(1)\right]  \,.
\]
Vice versa, if $\alpha>\alpha_{\mathrm{s}}(q)(1+o_{q}(1))$, such a graph is
with high probability uncolorable.
\end{theorem}

\begin{theorem}
\textrm{(Clustering of $q$-colorings \cite{AchlioCojas}) }%
\label{thm:qcolclust} The set of proper $q$-colorings of random graph with $n$
vertices and $n\alpha$ edges is clustered with high probability if
$\alpha>\alpha_{\mathrm{d}}(q)$, where
\[
\alpha_{\mathrm{d}}(q)=\frac{q}{2}\,[\log q+o(\log q)]\,.
\]

\end{theorem}

One of our main results is to prove a corresponding reconstruction theorem for
this model as follows.

\begin{theorem}
\label{thm:ColorRec} \textrm{(Graph $q$-coloring reconstruction)} Let
$\mu(\underline{x})$ be the uniform measure over of proper $q$-colorings of
random graph with $n$ vertices and $n\alpha$ edges. For $q$ large enough, the
reconstruction problem is solvable for $\mu$ if $\alpha>\alpha_{\mathrm{r}%
}(q)$, where
\[
\alpha_{\mathrm{r}}(k)=\frac{q}{2}\left[  \log q+\log\log q+O\left(  1\right)
\right]  \,.
\]
Vice versa, the reconstruction problem is unsolvable, with high probability,
if $\alpha<\alpha_{\mathrm{r}}(q)$.
\end{theorem}


\subsection{General strategy}

\label{sec:generalstrategy}

The results described in the previous section are of three types: bounds on
the satisfiability thresholds, cf. Proposition \ref{propo:SATCSP} and Theorem
\ref{thm:qcol}; on the clustering threshold, cf. Theorems \ref{thm:ClustCSP}
and \ref{thm:qcolclust}; on the reconstruction threshold, cf. Theorems
\ref{reconstruction} and \ref{thm:ColorRec}. The proof strategy is as follows.

\vspace{0.1cm}

\noindent\emph{The satisfiability threshold} can be upper bounded using the
first moment of the number of solutions, and lower bounded using the second
moment method. This technique is by now discussed in detail in
\cite{AchlioMoore,achlioassaf,ANP05}; we describe its application to the
general $\operatorname*{CSP}(n,\alpha,p)$ ensemble is done in Appendix~
\ref{sec:SecondMom}.

\vspace{0.1cm}

\noindent\emph{The clustering threshold} can be upper bounded through an
analysis of the recursive `whitening' process that associates to each cluster
a single configuration in an extended space \cite{AchlioRicci}. The improved
bounds in Theorems \ref{thm:ClustCSP} and \ref{thm:qcolclust} can be obtained
by approximating the CSP ensemble with an appropriate `planted' ensemble
\cite{AchlioCojas}. Since this approach is explained in detail in
\cite{AchlioCojas}, we will only present the various technical steps.

\vspace{0.1cm}

\noindent\emph{The reconstruction threshold} is characterized via a three-step procedure:

\vspace{0.05cm}

\noindent\textbf{(1)} Bound the reconstruction threshold for an appropriate
ensemble of (infinite) tree instances, i.e. CSP instances for which the
associated factor graph is an infinite Galton-Watson tree. In the case of
proper $q$-colorings, a sharp characterization was obtained independently by
two groups in the past year \cite{Bha-Ver-Vig,sly}. In Section
\ref{sec:TreeRec} we prove sharp bounds on tree reconstruction for binary
CSPs. The proof amounts to deriving an exact distributional recursion for the
so-called belief process, and carefully bounding its asymptotic behavior.

\noindent\textbf{(2)} Given two `balanced' solutions $\underline{x}^{(1)}$,
$\underline{x}^{(2)}$ (a solution is balanced if each possible variable value
is taken on the same number of vertices), define their \emph{joint type}
$\nu(x,y)$ as the matrix such that the fraction of vertices $i$ with
$x_{i}^{(1)}=x$ and $x_{i}^{(2)}=y$ is equal to $\nu(x,y)$. Consider the
number $Z_{\mathrm{b}}(\nu)$ of balanced solution pairs $\underline{x}_{1}$,
$\underline{x}_{2}$ with joint type $\nu$. One has to show that $\mathbb{E}%
\,Z_{\mathrm{b}}(\nu)$ is exponentially dominated by its value at the uniform
type $\overline{\nu}(x,y)=1/q^{2}$ (with $q=2$ for binary CSPs). More
precisely $\mathbb{E}\,Z_{\mathrm{b}}(\nu)\doteq\exp\{n\Phi(\nu)\}$ with
$\Phi$ achieving its unique maximum at $\overline{\nu}$.

This is also a crucial step in the second moment method. It was accomplished
in \cite{achlioassaf} for proper $q$-colorings of random graphs. In the case
of binary CSPs, we prove this estimate in Section \ref{sec:SecondMom}.

\noindent\textbf{(3)} Prove that the above imply that the set of solutions of
a random instance is, with high probability, \emph{roughly spherical}. By this
we mean that the joint type $\nu_{12}$ of two uniformly random solutions
$\underline{x}^{(1)},\underline{x}^{(2)}$ satisfies $||\nu_{12}-\overline{\nu
}||_{\operatorname*{TV}}\leq\delta$ with high probability for all $\delta>0$.
Notice that this implication requires bounding the expected ratio of
$Z_{\mathrm{b}}(\nu)$ to the total number of solution pairs. We prove that the
implication nevertheless holds in Section \ref{sec:ColorRec} for
$q$-colorings. The argument for binary CSP's is completely analogous, and we
omit it.

Finally, it was proved in \cite{mon-ger} that, under such a sphericity
condition, graph reconstruction and tree reconstruction are equivalent, which
finishes the proof of Theorems \ref{reconstruction} and \ref{thm:ColorRec}.

\vspace{0.2cm}

Notice that the techniques used for the clustering and reconstruction
thresholds are very different. Thus it is a surprising (and arguably deep)
phenomenon that they do coincide as far as the present techniques can tell.

\section{Tree ensembles and tree reconstruction for binary $k$-CSP ensembles}

\label{sec:TreeRec}

In this section we define tree ensembles and prove estimates about the
corresponding tree reconstruction thresholds.

\subsection{The tCSP$\left(  \alpha,p\right)  $ ensemble}

The ensemble $\operatorname*{tCSP}(\alpha,p)$ is defined by $\alpha
\in\mathbb{R}_{+}$ and a distribution $p$ over Boolean functions
$\varphi:\{-1,+1\}^{k}\rightarrow\{0,1\}$. We assume the conditions on the
distribution $p$ introduced in Section~\ref{sec:BinaryCSP}. An (infinite)
instance from this ensemble is generated starting by a root variable node
$\o $, drawing an integer $\eta\overset{\mathcal{D}}{=}$
$\operatorname*{Poisson}(k\alpha)$ and connecting $\o $ to $\eta$ function
nodes $\{1,\dots,\eta\}$. Each function node has degree $k$, and each of its
$k-1$ descendants is the root of an independent infinite tree. Finally, each
function node $a$ is associated independently, with a random clause $\varphi$
drawn according to $p$.

A uniform solution for such an instance is sampled by drawing the root value
$\mathbf{x}_{\o }\in\{-1,+1\}$ uniformly at random. The values of descendants
of each variable node $i$ are then drawn recursively. If the function node $a$
connects $i$ to $i_{1},\ldots,i_{k-1}$, then the values $\mathbf{x}_{i_{1}%
},\dots,\mathbf{x}_{i_{k}}$ are sampled uniformly from those that satisfy the
clause in $a$, that is, such that the quantity $\varphi\left(  x_{i},x_{i_{1}%
},\ldots,x_{i_{k-1}}\right)  $ is equal to $1$.

By the \emph{balance} condition, this procedure can be shown to be equivalent
to sampling a solution according to the `free boundary Gibbs measure.' The
latter is a distribution over solutions of the entire (infinite)
$\operatorname*{tCSP}$ formula defined by considering the unifom distribution
over solutions of the first $\ell$ generations of the tree, and then letting
$\ell\rightarrow\infty$.

\subsection{Reconstruction}

Given any fixed tree ensemble $T$, let $\mathbf{x}$ be a random satisfying
assignment for $T$ according to the distribution described previously. We
denote by $\mathbf{x}_{\ell}$ the value of $\mathbf{x}$ at the variables at
generation $\ell$, and in the case that the root degree is $1$, we denote by
$\mathbf{x}_{0,1},\ldots,\mathbf{x}_{0,k-1}$, the value at the variable nodes
connected to the unique child of the root. Also, we use $\eta_{0}$ for the
root degree of $T$. If the tree ensemble $T$ has root degree $\eta_{0}=d$, we
denote by $T_{i}$, $i=1,\ldots,d$, the subtree generated by the root, its
$i^{th}$ children and its descendents. If $\eta_{0}=1$, we denote by
$T_{i}^{\prime}$, $i=1,\ldots,k-1$, the subtree generated by the $i^{th} $
child of the root's child and its descendents.

Finally, because the tree ensemble $T$ could be random (for instance we denote
by $\mathbf{T}$ a random $\operatorname*{tCSP}\left(  \alpha,p\right)  $), we
will use\textbf{\ }$\boldsymbol{E}$ for expectation respect to $\mathbf{T} $,
and $\left\langle \cdot\right\rangle _{T}$ for expectation respect to
$\mathbf{x}$ (given $\mathbf{T}$) and $\mathbb{E}$ for expectation respect to
any other independent random variable (adding, if not in context, a subindex
to indicate such random variable).

\emph{Reconstruction:} For a fixed tree ensemble $T$, let $\mu_{_{\emptyset
,\ell}}$ be the joint distribution of $\left(  \mathbf{x}_{0},\mathbf{x}%
_{\ell}\right)  $ and let $\mu_{_{\emptyset}}$, $\mu_{_{\ell}}$ be the
marginal distribution of $\mathbf{x}_{0}$ and $\mathbf{x}_{\ell}$
respectively. The reconstruction rate for $T$ is defined as the quantity
$\left\Vert \mu_{\emptyset,\ell}\left(  \cdot,\cdot\right)  -\mu_{\emptyset
}\left(  \cdot\right)  \mu_{\ell}\left(  \cdot\right)  \right\Vert
_{\operatorname*{TV}}$. We say that the reconstruction problem for $T$ is
\emph{tree-solvable} if
\[
\liminf\limits_{\ell\rightarrow\infty}\left\Vert \mu_{\emptyset,\ell}\left(
\cdot,\cdot\right)  -\mu_{\emptyset}\left(  \cdot\right)  \mu_{\ell}\left(
\cdot\right)  \right\Vert _{\operatorname*{TV}}>0\text{.}%
\]
Analogously, if $\mathbf{T}$ is a random $\operatorname*{tCSP}\left(
\alpha,p\right)  $, we define the reconstruction rate of $\mathbf{T}$ as
$\mathbf{E}\left\Vert \mu_{\emptyset,\ell}\left(  \cdot,\cdot\right)
-\mu_{\emptyset}\left(  \cdot\right)  \mu_{\ell}\left(  \cdot\right)
\right\Vert _{\operatorname*{TV}}$, and we say that the reconstruction problem
for $\mathbf{T}$ is \emph{tree-solvable}
\[
\liminf\limits_{\ell\rightarrow\infty}\mathbf{E}\left\Vert \mu_{\emptyset
,\ell}\left(  \cdot,\cdot\right)  -\mu_{\emptyset}\left(  \cdot\right)
\mu_{\ell}\left(  \cdot\right)  \right\Vert _{\operatorname*{TV}}>0\text{.}%
\]

\emph{Bias, compatibility:} Given a satisfying assignment $x_{\ell}$ for the
variables at generation $\ell$, define the `bias' of the root, restricted to
the value of the variables at level $\ell$, as
\[
h_{T}\left(  x_{\ell}\right)  \overset{def}{=}\left\langle \mathbf{x}%
_{0}\left\vert \mathbf{x}_{\ell}=x_{\ell}\right.  \right\rangle _{T}\text{.}%
\]
Throughout the next proofs we will study $h_{T}\left(  x_{\ell}\right)  $, for
$x_{l}$ random and subject to different kind of distributions. Notice that
under the balance condition $\left\Vert \mu_{\emptyset,\ell}\left(
\cdot,\cdot\right)  -\mu_{\emptyset}\left(  \cdot\right)  \mu_{\ell}\left(
\cdot\right)  \right\Vert _{\operatorname*{TV}}=\left\langle \left\vert
h_{T}\left(  \mathbf{x}_{\ell}\right)  \right\vert \right\rangle _{T}$.

Now, let $D_{T}\left(  x_{\ell}\right)  \overset{def}{=}\left\{  x\right\}  $
if $h_{T}\left(  x_{\ell}\right)  =x$, $D_{T}\left(  x_{\ell}\right)
\overset{def}{=}\left\{  -1,1\right\}  $ if $\left\vert h_{T}\left(  x_{\ell
}\right)  \right\vert <1$. Observe that $D_{T}\left(  x_{\ell}\right)  $
consists of the values of the root that are compatible with the assignment
$x_{\ell}$ for the variables at generation $l$.

\emph{Domain of clauses:}\textit{\ }Given a binary function $\varphi\left(
x_{0},\ldots,x_{k-1}\right)  $, define the partial solution sets
\begin{align*}
&  S^{+}\left(  \varphi\right)  \overset{def}{=}\left\{  \left(
x_{1},,x_{k-1}\right)  :\varphi\left(  1,x_{1},\ldots,x_{k-1}\right)
=1\right\}  \text{,\qquad}\\
&  S^{-}\left(  \varphi\right)  \overset{def}{=}\left\{  \left(
x_{1},,x_{k-1}\right)  :\varphi\left(  -1,x_{1},\ldots,x_{k-1}\right)
=1\right\}  \text{,}%
\end{align*}

\[
\Lambda^{+}\left(  \varphi\right)  \overset{def}{=}S^{+}\left(  \varphi
\right)  \backslash S^{-}\left(  \varphi\right)  \text{,}\qquad\Lambda
^{-}\left(  \varphi\right)  \overset{def}{=}S^{-}\left(  \varphi\right)
\backslash S^{+}\left(  \varphi\right)
\]
If the clause $\varphi$ is balanced and feasible, we have that $\left\vert
S^{+}\left(  \varphi\right)  \right\vert =\left\vert S^{-}\left(
\varphi\right)  \right\vert =2^{k-1}\left\Vert \varphi\right\Vert ^{2}$ and
$\left\vert \Lambda^{+}\left(  \varphi\right)  \right\vert =\left\vert
\Lambda^{-}\left(  \varphi\right)  \right\vert =2^{k}\operatorname*{I}%
_{1}\left(  \varphi\right)  $.

\begin{theorem}
\label{thm:TreeRecoKCSP} The reconstruction problem for the ensemble
\textrm{tCSP}$(\alpha,p)$ is \emph{tree-solvable} if and only if
$\alpha>\alpha_{\mathrm{tree}}(k)$ where
\[
\alpha_{\mathrm{tree}}(k)=\frac{\Omega_{k}}{k}\,\{\log k+o(\log k)\}\,.
\]

\end{theorem}

\begin{proof}
\emph{Upper bound:}

Given a tree ensemble $T$, the rate of `naive reconstruction' for $T$ is
defined as%
\[
z_{\ell}\left(  T\right)  \overset{def}{=}\left\langle \mathbb{I}\left[
h_{T}\left(  \mathbf{x}_{\ell}\right)  =1\right]  \right\rangle _{T}\text{
(}=\left\langle \mathbb{I}\left[  h_{T}\left(  \mathbf{x}_{\ell}\right)
=-1\right]  \right\rangle _{T}\text{ by the balance condition),}%
\]
which indicates the probability that a random assignment for the variables at
generation $\ell$, distributed as $\mathbf{x}_{\ell}$, fixes the root to be
equal to $1$ (or $-1$). It is easy to see that $\left\langle \left\vert
h_{T}\left(  \mathbf{x}_{\ell}\right)  \right\vert \right\rangle _{T}\geq
z_{\ell}\left(  T\right)  $. Observe also, that for any $x,y\in\left\{
-1,1\right\}  $,
\begin{equation}
\left\langle \mathbb{I}\left[  h_{T}\left(  \mathbf{x}_{\ell}\right)
=x\right]  \left\vert \mathbf{x}_{0}=y\right.  \right\rangle _{T}=2z_{\ell
}\left(  T\right)  \delta_{x,y}\text{.} \label{f2}%
\end{equation}
Thus, our objective is to show that in an appropiate regime of the parameter
$\alpha$, the quantity $\mathbf{E}\left[  z_{\ell}\left(  \mathbf{T}\right)
\right]  $ remains bounded away from zero as $\ell\rightarrow\infty$, implying
tree-solvability of the reconstruction problem in such regime. Indeed, this
implies tree-solvability by `naive reconstruction', i.e. by the procedure that
assigns to the root any value compatible with the values at generation $\ell$.
By notational convenience, define
\[
z_{\ell}\left(  \alpha\right)  =2\mathbf{E}\left[  z_{\ell}\left(
\mathbf{T}\right)  \right]  \text{ and }\widehat{z_{\ell}}\left(
\alpha\right)  =2\mathbf{E}\left[  z_{\ell}\left(  \mathbf{T}\right)
\left\vert \eta_{0}=1\right.  \right]  \text{.}%
\]
Now, notice that for a tree ensemble $T$ with root degree $\eta_{0}=d$, and
any assignment $x_{\ell}$ for the variables at\ generation $\ell$,
$h_{T}\left(  x_{\ell}\right)  =1$ iff $h_{T}\left(  x_{\ell}\upharpoonright
T_{i}\right)  =1$ for some $i=1,\ldots,d$, so that
\begin{align*}
2z_{\ell}\left(  T\right)   &  =\left\langle 1-%
{\displaystyle\prod\limits_{i=1}^{d}}
\left(  1-\mathbb{I}\left[  h_{T_{i}}\left(  \mathbf{x}_{\ell}\upharpoonright
T_{i}\right)  =1\right]  \right)  \left\vert \mathbf{x}_{0}=1\right.
\right\rangle _{T}\\
&  =1-%
{\displaystyle\prod\limits_{i=1}^{d}}
\left\langle \left(  1-\mathbb{I}\left[  h_{T_{i}}\left(  \mathbf{x}_{\ell
}\right)  =1\right]  \right)  \Big\vert \mathbf{x}_{0}=1 \right\rangle
_{T_{i}}\text{ (By the tree Markov property)}\\
&  =1-%
{\displaystyle\prod\limits_{i=1}^{d}}
\left(  1-2z_{\ell}\left(  T_{i}\right)  \right)\,.
\end{align*}
Therefore, averaging over $T$, we get
\begin{align*}
z_{\ell}\left(  \alpha\right)   &  =\mathbb{E}_{\eta}\left[  1-%
{\displaystyle\prod\limits_{i=1}^{\eta}}
\left(  1-\widehat{z_{\ell}}\left(  \alpha\right)  \right)  \right]  \text{,
}\eta\sim\operatorname*{Poisson}\left(  k\alpha\right) \\
&  =1-\exp\left(  -k\alpha\widehat{z_{\ell}}\left(  \alpha\right)  \right)
\text{.}%
\end{align*}
On the other hand, given a tree ensemble $T$ with root degree $\eta_{0}=1$ and
with the clause $\varphi$ assigned to the root's child, we have that for any
satisfying assignment $x_{\ell}$ for the variables at generation $\ell$,
$h_{T}\left(  x_{\ell}\right)  =1$ iff
\begin{equation}%
{\displaystyle\prod\limits_{i=1}^{k-1}}
D_{T_{i}^{\prime}}\left(  x_{\ell-1}^{\left(  i\right)  }\right)
\subseteq\Lambda^{+}\left(  \varphi\right)  \text{,} \label{f1}%
\end{equation}
where $x_{\ell-1}^{\left(  i\right)  }$ is the assignment $x_{\ell
}\upharpoonright T_{i}^{\prime}$ for the variables at generation $\ell-1$ in
the subtree $T_{i}^{\prime}$. Observe that (\ref{f1}) holds, in particular, if
for some $a=\left(  a_{1},\ldots,a_{k-1}\right)  \in\Lambda^{+}\left(
\varphi\right)  $, $h_{T_{i}^{\prime}}\left(  x_{\ell-1}^{\left(  i\right)
}\right)  =a_{i}$ for $i=1,\ldots,k-1$. Therefore, if $\mathbf{y}=\left(
\mathbf{y}_{1},\ldots,\mathbf{y}_{k-1}\right)  $ denotes a random uniform
vector from $S^{+}\left(  \varphi\right)  $, we have
\begin{align*}
z_{\ell}\left(  T\right)   &  \geq\frac{1}{2}%
{\displaystyle\sum\limits_{a\in\Lambda^{+}\left(  \varphi\right)  }}
\left\langle
{\displaystyle\prod\limits_{i=1}^{k-1}}
\mathbb{I}\left[  h_{T_{i}^{\prime}}\left(  \mathbf{x}_{\ell-1}^{\left(
i\right)  }\right)  =a_{i}\right]  \left\vert \mathbf{x}_{0}=1\right.
\right\rangle _{T}\\
&  =\frac{1}{2}%
{\displaystyle\sum\limits_{a\in\Lambda^{+}\left(  \varphi\right)  }}
\mathbb{E}_{\mathbf{y}}%
{\displaystyle\prod\limits_{i=1}^{k-1}}
\left\langle \mathbb{I}\left[  h_{T_{i}^{\prime}}\left(  \mathbf{x}_{\ell
-1}\right)  =a_{i}\right]  \left\vert \mathbf{x}_{0}=y_{i}\right.
\right\rangle _{T_{i}^{\prime}}\text{ (By the tree Markov property)}\\
&  =\frac{\left\vert \Lambda^{+}\left(  \varphi\right)  \right\vert
}{\left\vert S^{+}\left(  \varphi\right)  \right\vert }%
{\displaystyle\prod\limits_{i=1}^{k-1}}
2z_{\ell-1}\left(  T_{i}^{\prime}\right)  \text{ (By Eq. (\ref{f2})),}%
\end{align*}
which implies, after averaging over $T$, that%
\[
\widehat{z_{\ell}}\left(  \alpha\right)  \geq\mathbb{E}_{\mathbf{\varphi}%
}\left[  \frac{2\operatorname*{I}_{1}\left(  \varphi\right)  }{\left\Vert
\varphi\right\Vert ^{2}}\right]  \left(  z_{\ell-1}\left(  \alpha\right)
\right)  ^{k-1}=\frac{\left(  z_{\ell-1}\left(  \alpha\right)  \right)
^{k-1}}{\Omega_{k}}\text{,}%
\]
which leads to the recursion $z_{\ell}\left(  \alpha\right)  \geq1-\exp\left(
-k\alpha\left(  z_{\ell-1}\left(  \alpha\right)  \right)  ^{k-1}/\Omega
_{k}\right)  $. Now, it is standard to verify that this recursion implies that
$z_{\ell}\left(  \alpha\right)  $ is, for all $\ell$, greater or equal than
the maximum of the fixed points of the function $g\left(  z\right)
=1-\exp\left(  -k\alpha z^{k-1}/\Omega_{k}\right)  $ in the interval $\left[
0,1\right]  $. The minimum value of $\alpha$ for which such fixed point is
positive is given by
\[
\alpha^{\ast}=\frac{\Omega_{k}\left(  1+u\left(  1+\frac{1}{u}\right)
^{k-2}\right)  }{k\left(  k-1\right)  }\text{,}%
\]
where $u$ is the unique solution of the equation $u=\left(  k-1\right)
\log\left(  1+u\right)  $. In particular, asymptotically in $k$, we have that
$\alpha^{\ast}=$ $\frac{\Omega_{k}}{k}\left(  \log k+\operatorname*{o}\left(
\log k\right)  \right)  $, which implies the upper bound for $\alpha
_{\operatorname*{tree}}$.

\emph{Lower bound}:

The matching lower bound on $\alpha_{\mathrm{tree}}(k)$ requires a more
elaborate proof; we first prove three lemmas, before returning to complete the
lower bound proof.
\end{proof}

Given a tree ensemble $T$, let $\mathbf{x}_{\ell}^{+}\overset{\mathcal{D}}%
{=}\left(  \mathbf{x}_{\ell}\left\vert \mathbf{x}_{0}=1\right.  \right)  $ and
$\mathbf{x}_{\ell}^{-}\overset{\mathcal{D}}{=}\left(  \mathbf{x}_{\ell
}\left\vert \mathbf{x}_{0}=-1\right.  \right)  $. When the tree ensemble is
not clear in the definition of $\mathbf{x}_{\ell}^{+}$ (or $\mathbf{x}_{\ell
}^{-}$), we add a subindex indicating the tree ensemble from where it is
defined. Notice that, if $\mu^{+}$ and $\mu^{-}$ are the distributions of
$\mathbf{x}_{\ell}^{+}$ and $\mathbf{x}_{\ell}^{-}$ respectively, then
\begin{equation}
\frac{d\mu^{-}}{d\mu^{+}}=\frac{1-h_{T}\left(  x_{l}\right)  }{1+h_{T}\left(
x_{l}\right)  }\text{.} \label{f4}%
\end{equation}
By the balance condition, it's clear that
\begin{equation}
h_{T}\left(  \mathbf{x}_{\ell}^{+}\right)  \overset{\mathcal{D}}{=}%
-h_{T}\left(  \mathbf{x}_{\ell}^{-}\right)  \text{.} \label{f3}%
\end{equation}
Also, it is easy to show that $\left\langle h_{T}\left(  \mathbf{x}_{\ell}%
^{+}\right)  \right\rangle _{T}=\left\langle \left[  h_{T}\left(
\mathbf{x}_{\ell}\right)  \right]  ^{2}\right\rangle _{T}$ (and therefore
$\left[  R_{l}\left(  T\right)  \right]  ^{2}\leq\left\langle h_{T}\left(
\mathbf{x}_{\ell}^{+}\right)  \right\rangle _{T}\leq R_{l}\left(  T\right)
$), so that non-reconstructibility for $T$ is equivalent to the condition
$\lim\limits_{\ell\rightarrow\infty}\left\langle h_{T}\left(  \mathbf{x}%
_{\ell}^{+}\right)  \right\rangle _{T}=0$ (see \cite{TreeRec}). Similarly, if
$\mathbf{T}$ is a random $\operatorname*{tCSP}\left(  \alpha,p\right)  $
ensemble, non-reconstructibility for $\mathbf{T}$, is equivalent to the
condition $\lim\limits_{\ell\rightarrow\infty}\mathbf{E}\left[  \left\langle
h_{\mathbf{T}}\left(  \mathbf{x}_{\ell}^{+}\right)  \right\rangle
_{\mathbf{T}}\right]  =0$.

\begin{lemma}
\label{lemma:Recursion} \textbf{(a)} Given a tree ensemble $T$ with root
degree $\eta_{0}=d$, we have
\begin{equation}
\left[  \frac{1-h_{T}\left(  \mathbf{x}_{\ell}^{+}\right)  }{1+h_{T}\left(
\mathbf{x}_{\ell}^{+}\right)  }\right]  \overset{\mathcal{D}}{=}\prod
_{i=1}^{d}\left[  \frac{1-h_{l,i}}{1+h_{l,i}}\right]  \text{,}
\label{eq:Recursion1}%
\end{equation}
where $\left(  h_{l,i}\right)  _{i=1}^{d}$ are independent random variables
such that $h_{l,i}\overset{\mathcal{D}}{=}h_{T_{i}}\left(  \mathbf{x}_{\ell
}^{+}\right)  $.

\textbf{(b)} Given a tree ensemble $T$ with root degree $\eta_{0}=1$ and with
the clause $\varphi$ assigned to the unique child of the root, we have that
\begin{equation}
\left[  \frac{1-h_{T}\left(  \mathbf{x}_{\ell+1}^{+}\right)  }{1+h_{T}\left(
\mathbf{x}_{\ell+1}^{+}\right)  }\right]  \overset{\mathcal{D}}{=}%
\frac{\operatorname*{T}_{h_{l}}\varphi(-1,\mathbf{s})}{\operatorname*{T}%
_{h_{l}}\varphi(1,\mathbf{s})}\text{,} \label{eq:Recursion2}%
\end{equation}
where $\mathbf{s\sim}\operatorname*{Unif}\left(  S^{+}\left(  \varphi\right)
\right)  $ and $h_{l}=\left(  h_{l,i}\right)  _{i=1}^{k-1}$ are independent
random variables such that $h_{l,i}\overset{\mathcal{D}}{=}h_{T_{i}^{\prime}%
}\left(  \mathbf{x}_{l}^{+}\right)  $.
\end{lemma}

\begin{proof}
This recursion follows straightforwardly from the recursive definition of tree
formulae. The balance condition on clauses implies
\[
\frac{1-h_{T}\left(  \mathbf{x}_{l}^{+}\right)  }{1+h_{T}\left(
\mathbf{x}_{l}^{+}\right)  }=\frac{\left\langle \mathbb{I}\left[
\mathbf{x}_{l}=\mathbf{x}_{l}^{+}\right]  \left\vert \mathbf{x}_{0}=-1\right.
\right\rangle _{T}}{\left\langle \mathbb{I}\left[  \mathbf{x}_{l}%
=\mathbf{x}_{l}^{+}\right]  \left\vert \mathbf{x}_{0}=1\right.  \right\rangle
_{T}}\text{.}%
\]
Therefore, if the root degree of $T$ is $\eta_{0}=d$, we have by the tree
Markov property that
\[
\frac{1-h_{T}\left(  \mathbf{x}_{l}^{+}\right)  }{1+h_{T}\left(
\mathbf{x}_{l}^{+}\right)  }=%
{\displaystyle\prod\limits_{i=1}^{d}}
\frac{\left\langle \mathbb{I}\left[  \mathbf{x}_{l}=\mathbf{x}_{l}%
^{+}\upharpoonright T_{i}\right]  \left\vert \mathbf{x}_{0}=-1\right.
\right\rangle _{T_{i}}}{\left\langle \mathbb{I}\left[  \mathbf{x}%
_{l}=\mathbf{x}_{l}^{+}\upharpoonright T_{i}\right]  \left\vert \mathbf{x}%
_{0}=1\right.  \right\rangle _{T_{i}}}\text{,}%
\]
and the last expression has the same distribution as $%
{\displaystyle\prod\limits_{i=1}^{d}}
\frac{1-u_{l,i}}{1+u_{l,i}}$, due to the fact that $\left(  \mathbf{x}_{l}%
^{+}\upharpoonright T_{i}\right)  _{i=1}^{d}$ are independent random
assignments for the variables at generation $l$ of $T_{i}$, such that
$\mathbf{x}_{l}^{+}\upharpoonright T_{i}\overset{\mathcal{D}}{=}%
\mathbf{x}_{l,T_{i}}^{+}$. This proves Eq.~(\ref{eq:Recursion1}). Now, if the
root degree of $T$ is $\eta_{0}=1$, define $\left(  \widetilde{\mathbf{x}%
}_{l,i}^{+}\right)  _{i=1}^{k-1}$ to be independent random assignments for the
variables at generation $l$ of the subtrees $T_{i}^{\prime}$, such that
$\widetilde{\mathbf{x}}_{l,i}^{+}\overset{\mathcal{D}}{=}\mathbf{x}%
_{l,T_{i}^{\prime}}^{+}$. By the tree Markov property, we have that $\left(
\mathbf{x}_{l+1}^{+}\upharpoonright T_{i}^{\prime}\right)  _{i=1}%
^{k-1}\overset{\mathcal{D}}{=}\left(  \mathbf{s}_{i}\widetilde{\mathbf{x}%
}_{l,i}^{+}\right)  _{i=1}^{k-1}$ where $\mathbf{s}\sim\operatorname*{Unif}%
S^{+}\left(  \varphi\right)  $. Using once more the tree Markov property, we
get
\begin{align*}
\left[  \frac{1-h_{T}\left(  \mathbf{x}_{\ell+1}^{+}\right)  }{1+h_{T}\left(
\mathbf{x}_{\ell+1}^{+}\right)  }\right]   &  =\frac{%
{\displaystyle\sum\limits_{y}}
\varphi\left(  -1,y\right)
{\displaystyle\prod\limits_{i=1}^{k-1}}
\left\langle \mathbb{I}\left[  \mathbf{x}_{l}=\mathbf{s}_{i}\widetilde
{\mathbf{x}}_{l,i}^{+}\right]  \left\vert \mathbf{x}_{0}=y_{i}\right.
\right\rangle _{T_{i}^{\prime}}}{%
{\displaystyle\sum\limits_{y}}
\varphi\left(  -1,y\right)
{\displaystyle\prod\limits_{i=1}^{k-1}}
\left\langle \mathbb{I}\left[  \mathbf{x}_{l}=\mathbf{s}_{i}\widetilde
{\mathbf{x}}_{l,i}^{+}\right]  \left\vert \mathbf{x}_{0}=y_{i}\right.
\right\rangle _{T_{i}^{\prime}}}\\
&  =\frac{\operatorname*{T}_{h_{l}}\varphi\left(  -1,\mathbf{s}\right)
}{\operatorname*{T}_{h_{l}}\varphi\left(  1,\mathbf{s}\right)  }\text{,}%
\end{align*}
which is precisely Eq.~(\ref{eq:Recursion2}).
\end{proof}

The first step of the above recursion can be analyzed exactly.

\begin{lemma}
\label{lemma:FirstStep} If $\mathbf{T}$ is a random $\operatorname*{tCSP}%
\left(  \alpha,p\right)  $ ensemble, then the random variable $h_{\mathbf{T}%
}\left(  \mathbf{x}_{1}^{+}\right)  $ takes values in $\{0,1\}$ and, if
$\alpha<(1-\delta)(\Omega_{k}\log k)/k$, we have $\mathbf{E}\,h_{\mathbf{T}%
}\left(  \mathbf{x}_{1}^{+}\right)  \leq1-k^{-1+\delta}$.
\end{lemma}



\begin{proof}
If $T$ is a tree ensemble with root degree $\eta_{0}=1$ and clause $\varphi$
assigned to the root's child, from the part b of lemma \ref{lemma:Recursion},
we have that $\frac{1-h_{T}\left(  \mathbf{x}_{1}^{+}\right)  }{1+h_{T}\left(
\mathbf{x}_{1}^{+}\right)  }\overset{\mathcal{D}}{=}\varphi\left(
-1,\mathbf{s}\right)  $ where $\mathbf{s}\sim\operatorname*{Unif}\left(
S^{+}\left(  \varphi\right)  \right)  $ (notice that $h_{0,i}\equiv1$).
Therefore, it follows that $h_{T}\left(  \mathbf{x}_{1}^{+}\right)  =1$ w.p.
$\frac{\left\vert \Lambda^{+}\left(  \varphi\right)  \right\vert }{\left\vert
S^{+}\left(  \varphi\right)  \right\vert }=1/\Omega_{k}$ and $h_{T}\left(
\mathbf{x}_{1}^{+}\right)  =0$ otherwise. Therefore, if $T$ is a tree ensemble
with root degree $\eta_{0}=d$, it follows from the part a of lemma
\ref{lemma:Recursion} that $h_{T}\left(  \mathbf{x}_{1}^{+}\right)  =1$ w.p.
$1-\left(  1-1/\Omega_{k}\right)  ^{d}$ and $h_{T}\left(  \mathbf{x}_{1}%
^{+}\right)  =0$ otherwise. This implies then that $h_{\mathbf{T}}\left(
\mathbf{x}_{1}^{+}\right)  $ is supported at $\left\{  0,1\right\}  $ and
$\mathbf{E}h_{\mathbf{T}}\left(  \mathbf{x}_{1}^{+}\right)  =1-\exp\left(
-k\alpha\left(  1-1/\Omega_{k}\right)  \right)  $. The conclusion follows
straightforwardly.

\end{proof}

For subsequent steps we track the averages, $h_{\ell}^{\operatorname*{ave}%
}\overset{def}{=}\mathbf{E}\,\left\langle h_{\mathbf{T}}\left(  \mathbf{x}%
_{l}^{+}\right)  \right\rangle _{\mathbf{T}}$ and $\widehat{h}_{\ell
}^{\operatorname*{ave}}\overset{def}{=}\mathbf{E}\,\left[  \left\langle
h_{\mathbf{T}}\left(  \mathbf{x}_{l}^{+}\right)  \right\rangle _{\mathbf{T}%
}\left\vert \eta_{0}=1\right.  \right]  $, using the following bounds.

\begin{lemma}
\label{lemma:IterBound} For any $\ell\geq0$ we have
\begin{align}
&  h_{\ell}^{\operatorname*{ave}}\leq1-e^{-2k\alpha\widehat{h}_{\ell
}^{\operatorname*{ave}}}\text{,}\,\;\;\;\;\;\widehat{h}_{\ell+1}%
^{\operatorname*{ave}}\leq\frac{1}{2}\,F_{k}(h_{\ell}^{\operatorname*{ave}%
})+\frac{1}{2}\,R_{k}(\sqrt{h_{\ell}^{\operatorname*{ave}}})\,\text{,}\\
&  F_{k}(\theta)\overset{def}{=}2\mathbb{E}_{\varphi}\left[  \frac
{(\varphi^{\left(  1\right)  },\operatorname*{T}_{\theta}\,\varphi^{\left(
1\right)  })}{\left\Vert \varphi\right\Vert ^{2}}\,\right]  \text{,}%
\;\;\;\;\;R_{k}(\theta)\overset{def}{=}2\mathbb{E}_{\varphi}{}_{\mathbf{i}%
}\left[  \frac{2\operatorname*{I}_{1}\left(  \varphi\right)  }{\left\Vert
\varphi\right\Vert ^{2}}\sum_{Q\subseteq\lbrack k-1]}|\left(  \varphi^{\left(
1\right)  },\gamma_{Q}\right)  |\,\theta^{\max(|Q|,2)}\right]  \text{,}%
\end{align}
Finally, if $h_{\ell}$ is supported on non-negative values, then
\begin{equation}
\widehat{h}_{\ell}^{\operatorname*{ave}}\leq F_{k}(h_{\ell}%
^{\operatorname*{ave}})\,\text{.} \label{eq:ineq3}%
\end{equation}

\end{lemma}

\begin{proof}
We will say that a random variable $\mathbf{X}\in\lbrack-1,+1]$ is
`consistent,' if $\mathbb{E}\,f(-\mathbf{X})=\mathbb{E}\left[  \left(
\frac{1-\mathbf{X}}{1+\mathbf{X}}\right)  \,f(\mathbf{X})\right]  $for every
function $f$ such that the expectation values exist. A useful preliminary
remark \cite{MezardMontanari} is that the random variable $h_{T}\left(
\mathbf{x}_{l}^{+}\right)  $ is consistent (no matter the tree ensemble). In
fact, this follows directly from the Eqs. (\ref{f4}) and (\ref{f3}) above. A
number of properties of consistent random variables can be found in
\cite{Rudiger}. Let us now consider the first inequality. If $T$ is a tree
ensemble with root degree $\eta_{0}=d$, it is immediate to from
Eq.~(\ref{eq:Recursion1}), that
\[
\left\langle \left(  \frac{1-h_{T}\left(  \mathbf{x}_{l}^{+}\right)  }%
{1+h_{T}\left(  \mathbf{x}_{l}^{+}\right)  }\right)  ^{1/2}\right\rangle
_{T}=\prod_{i=1}^{d}\left\langle \left(  \frac{1-h_{T_{i}}\left(
\mathbf{x}_{l}^{+}\right)  }{1+h_{T_{i}}\left(  \mathbf{x}_{l}^{+}\right)
}\right)  ^{1/2}\right\rangle _{T_{i}}\,\text{.}%
\]
It is possible to show that consistency implies $E\,X=E\,X^{2}$ and $E\left(
\frac{1-X}{1+X}\right)  ^{1/2}=E\sqrt{1-X^{2}}$ (through the test functions
$f\left(  x\right)  =x\left(  1+x\right)  $ and $f\left(  x\right)  =x\left(
1+x\right)  ^{1/2}\left(  1-x\right)  ^{-1/2}$), we thus have
\[
\sqrt{1-\left\langle h_{T}\left(  \mathbf{x}_{l}^{+}\right)  \right\rangle
_{T}}\geq\left\langle \sqrt{1-\left[  h_{T}\left(  \mathbf{x}_{l}^{+}\right)
\right]  ^{2}}\right\rangle _{T}=\prod_{i=1}^{d}\left\langle \sqrt{1-\left[
h_{T_{i}}\left(  \mathbf{x}_{l}^{+}\right)  \right]  ^{2}}\right\rangle
_{T_{i}}\geq\prod_{i=1}^{d}\left(  1-\left\langle h_{T_{i}}\left(
\mathbf{x}_{l}^{+}\right)  \right\rangle _{T_{i}}\right)  \text{.}%
\]
This implies in particular, if $\mathbf{T}$ is a random $\operatorname*{tCSP}%
\left(  \alpha,p\right)  $,
\[
\sqrt{1-\mathbf{E}\left\langle h_{\mathbf{T}}\left(  \mathbf{x}_{l}%
^{+}\right)  \right\rangle _{\mathbf{T}}}\geq\mathbb{E}_{\eta}\left[
\prod_{i=1}^{\eta}\left(  1-\mathbf{E}\left[  \left\langle h_{\mathbf{T}%
}\left(  \mathbf{x}_{l}^{+}\right)  \right\rangle _{\mathbf{T}}\left\vert
\eta_{0}=1\right.  \right]  \right)  \right]  \text{, }\eta\sim
\operatorname*{Poisson}\left(  k\alpha\right)  \text{,}%
\]
from where the first inequality follows.

Now, from the recursion Eq. (\ref{eq:Recursion2}), we have for a tree ensemble
$T$ with root degree $\eta_{0}=1$, and random clause $\varphi$ assigned to the
child of the root,%
\[
h_{T}\left(  \mathbf{x}_{l+1}^{+}\right)  =\frac{2\operatorname*{T}_{h_{l}%
}\varphi^{\left(  1\right)  }\left(  \mathbf{s}\right)  }{1+\operatorname*{T}%
_{h_{l}}\psi\left(  \mathbf{s}\right)  }\text{,\quad\quad}\psi\left(
s\right)  \overset{def}{=}\varphi\left(  1,s\right)  \varphi\left(
-1,s\right)
\]
or alternatively,
\[
h_{T}\left(  \mathbf{x}_{l+1}^{+}\right)  =\operatorname*{T}\nolimits_{h_{l}%
}\varphi^{\left(  1\right)  }\left(  \mathbf{s}\right)  +\left(
\operatorname*{T}\nolimits_{h_{l}}\varphi^{\left(  1\right)  }\left(
\mathbf{s}\right)  \right)  \mathcal{G}_{k}\left(  h_{l},\mathbf{s}\right)
\text{,\quad\quad}\mathcal{G}_{k}\left(  h_{l},s\right)  \overset{def}%
{=}\left[  \frac{1-\operatorname*{T}_{h_{l}}\psi\left(  s\right)
}{1+\operatorname*{T}_{h_{l}}\psi\left(  s\right)  }\right]  \text{,}%
\]
where $\mathbf{s\sim}\operatorname*{Unif}S^{+}\left(  \varphi\right)  $.
Notice that for any antisymmetric function $f\left(  s\right)  $, we have that
$\mathbb{E}_{\mathbf{s}}f\left(  \mathbf{s}\right)  =\frac{\left(
\varphi^{\left(  1\right)  },f\right)  }{\left\Vert \varphi\right\Vert ^{2}}$.
Therefore, due to the fact that $\operatorname*{T}\nolimits_{h_{l}}%
\varphi^{\left(  1\right)  }\left(  s\right)  $ is antisymmetric and
$\mathcal{G}_{k}\left(  h_{l},s\right)  $ is symmetric (both in $s$ and
$h_{l}$, actually), we have the formulas
\begin{equation}
\left\langle h_{T}\left(  \mathbf{x}_{l+1}^{+}\right)  \right\rangle
_{T}=\frac{2}{\left\Vert \varphi\right\Vert ^{2}}\left\langle \left(
\varphi^{\left(  1\right)  },\frac{\operatorname*{T}_{h_{l}}\varphi^{\left(
1\right)  }\left(  \mathbf{s}\right)  }{1+\operatorname*{T}_{h_{l}}\psi\left(
\mathbf{s}\right)  }\right)  \right\rangle _{T} \label{f8}%
\end{equation}
and
\begin{equation}
\left\langle h_{T}\left(  \mathbf{x}_{l+1}^{+}\right)  \right\rangle
_{T}=\left\langle \frac{\left(  \varphi^{\left(  1\right)  },\operatorname*{T}%
_{h_{l}}\varphi^{\left(  1\right)  }\right)  }{\left\Vert \varphi\right\Vert
^{2}}\right\rangle _{T}+\left\langle \frac{\left(  \varphi^{\left(  1\right)
},\left(  \operatorname*{T}\nolimits_{h_{l}}\varphi^{\left(  1\right)
}\right)  \mathcal{G}_{k}\left(  h_{l},\mathbf{\cdot}\right)  \right)
}{\left\Vert \varphi\right\Vert ^{2}}\right\rangle _{T}\text{.} \label{f9}%
\end{equation}
In the last expression, the first term is equal to $\frac{\left(
\varphi^{\left(  1\right)  },\operatorname*{T}_{\left\langle h_{l}%
\right\rangle _{T}}\varphi^{\left(  1\right)  }\right)  }{\left\Vert
\varphi\right\Vert ^{2}}$, while the second term can be writen, using Fourier
expansion, as
\[
\frac{1}{\left\Vert \varphi\right\Vert ^{2}}%
{\displaystyle\sum\limits_{\substack{Q\subseteq\left[  k-1\right]
\\\left\vert Q\right\vert \text{ odd}}}}
\left(  \varphi^{\left(  1\right)  },\gamma_{Q}\mathbb{E}_{h_{l}}\left[
\gamma_{Q}\left(  h_{l}\right)  \mathcal{G}_{k}\left(  h_{l},\mathbf{\cdot
}\right)  \right]  \right)  \left(  \varphi^{\left(  1\right)  },\gamma
_{Q}\right)  \text{.}%
\]
Using the fact that $\mathbb{E}\left\vert \mathbf{X}\right\vert \leq\left(
\mathbb{E}\mathbf{X}\right)  ^{1/2}$ for consistent random variables, we can
bound the terms with $\left\vert Q\right\vert \geq3$ by
\[
\frac{\left\vert \left(  \varphi^{\left(  1\right)  },1\right)  \right\vert
}{\left\Vert \varphi\right\Vert ^{2}}%
{\displaystyle\sum\limits_{\substack{Q\subseteq\left[  k-1\right]
\\\left\vert Q\right\vert \geq3\text{ odd}}}}
\left\vert \left(  \varphi^{\left(  1\right)  },\gamma_{Q}\right)  \right\vert
\left(
{\displaystyle\prod\limits_{i\in Q}}
\left\langle h_{T_{i}}\left(  \mathbf{x}_{l}^{+}\right)  \right\rangle
_{T_{i}}\right)  ^{1/2}\text{.}%
\]
Also, using the fact that for any even function $f\left(  x\right)  $ with
$0\leq f\left(  x\right)  \leq1$ and a consistent random variable $\mathbf{X}%
$, we have
\[
|\mathbb{E}[\mathbf{X}f(\mathbf{X})]|=|\mathbb{E}[2\mathbf{X}^{2}%
f(\mathbf{X})/(1+\mathbf{X})\mathbb{I}_{\{\mathbf{X}\geq0\}}]|\leq
|\mathbb{E}[2\mathbf{X}^{2}/(1+\mathbf{X})\mathbb{I}_{\{\mathbf{X}\geq
0\}}]|=|\mathbb{E}[\mathbf{X}]|\text{,}%
\]
we can bound the terms with $\left\vert Q\right\vert =1$, by
\[
\frac{\left\vert \left(  \varphi^{\left(  1\right)  },1\right)  \right\vert
}{\left\Vert \varphi\right\Vert ^{2}}%
{\displaystyle\sum\limits_{i=1}^{k-1}}
\left(  \varphi^{\left(  1\right)  },\gamma_{\left\{  i\right\}  }\right)
\left\vert \left\langle h_{T_{i}}\left(  \mathbf{x}_{l}^{+}\right)
\right\rangle _{T_{i}}\right\vert \text{.}%
\]
Therefore, for a random $\operatorname*{tCSP}\left(  \alpha,p\right)  $ with
root degree $\eta_{0}=1$, we obtain after averaging
\[
\widehat{h}_{l+1}^{\operatorname*{ave}}\leq\mathbb{E}_{\varphi}\frac{\left(
\varphi^{\left(  1\right)  },\operatorname*{T}_{h_{l}^{\operatorname*{ave}}%
}\varphi^{\left(  1\right)  }\right)  }{\left\Vert \varphi\right\Vert ^{2}%
}+\mathbb{E}_{\varphi}\left[  \frac{2\operatorname*{I}_{1}\left(
\varphi\right)  }{\left\Vert \varphi\right\Vert ^{2}}%
{\displaystyle\sum\limits_{\substack{Q\subseteq\left[  k-1\right]
\\\left\vert Q\right\vert \geq3\text{ odd}}}}
\left\vert \left(  \varphi^{\left(  1\right)  },\gamma_{Q}\right)  \right\vert
\left(  \sqrt{h_{l}^{\operatorname*{ave}}}\right)  ^{\max\left\{  \left\vert
Q\right\vert ,2\right\}  }\right]  \text{,}%
\]
which is precisely the second inequality in the Lemma.

Now, suppose that $h_{l}$ is supported on non-negative values and let
$A_{s}=\left\{  h_{l}:\operatorname*{T}_{h_{l}}\varphi^{\left(  1\right)
}\left(  s\right)  >0\right\}  $. Notice that the complement of $A_{s}$ is
$-A_{s}$ (due to the antisymmetry of $\operatorname*{T}_{h_{l}}\varphi
^{\left(  1\right)  }\left(  s\right)  $ respect to $h_{l}$). Therefore, using
the consistency of the random variables $h_{l,i}$, from the Eq. (\ref{f8}) we
get
\begin{align*}
\left\langle h_{T}\left(  \mathbf{x}_{l+1}^{+}\right)  \right\rangle _{T}  &
=\frac{2}{\left\Vert \varphi\right\Vert ^{2}}\left\langle \left(
\varphi^{\left(  1\right)  },\frac{\operatorname*{T}_{h_{l}}\varphi^{\left(
1\right)  }\left(  \mathbf{s}\right)  }{1+\operatorname*{T}_{h_{l}}\psi\left(
\mathbf{s}\right)  }\right)  \mathbb{I}\left(  h_{l}\in A_{\mathbf{s}}\right)
-\left(  \varphi^{\left(  1\right)  },\frac{\operatorname*{T}_{-h_{l}}%
\varphi^{\left(  1\right)  }\left(  \mathbf{s}\right)  }{1+\operatorname*{T}%
_{-h_{l}}\psi\left(  \mathbf{s}\right)  }\right)  \mathbb{I}\left(  -h_{l}\in
A_{\mathbf{s}}\right)  \right\rangle _{T}\\
&  =\frac{2}{\left\Vert \varphi\right\Vert ^{2}}\left\langle \left(
\varphi^{\left(  1\right)  },\frac{\operatorname*{T}_{h_{l}}\varphi^{\left(
1\right)  }\left(  \mathbf{s}\right)  }{1+\operatorname*{T}_{h_{l}}\psi\left(
\mathbf{s}\right)  }\right)  \mathbb{I}\left(  h_{l}\in A_{\mathbf{s}}\right)
\left[  1-%
{\displaystyle\prod\limits_{i=1}^{k-1}}
\frac{1-h_{l,i}}{1+h_{l,i}}\right]  \right\rangle _{T}\\
&  \leq\frac{2}{\left\Vert \varphi\right\Vert ^{2}}\left\langle \left(
\varphi^{\left(  1\right)  },\operatorname*{T}\nolimits_{h_{l}}\varphi
^{\left(  1\right)  }\left(  \mathbf{s}\right)  \right)  \mathbb{I}\left(
h_{l}\in A_{\mathbf{s}}\right)  \left[  1-%
{\displaystyle\prod\limits_{i=1}^{k-1}}
\frac{1-h_{l,i}}{1+h_{l,i}}\right]  \right\rangle _{T}\\
&  =\frac{2\left(  \varphi^{\left(  1\right)  },\operatorname*{T}%
\nolimits_{\left\langle h_{l}\right\rangle _{T}}\varphi^{\left(  1\right)
}\left(  \mathbf{s}\right)  \right)  }{\left\Vert \varphi\right\Vert ^{2}%
}\text{.}%
\end{align*}
Therefore, for a random $\operatorname*{tCSP}\left(  \alpha,p\right)  $ with
root degree $\eta_{0}=1$, we obtain after averaging, that
\[
\widehat{h}_{l+1}^{\operatorname*{ave}}\leq2\mathbb{E}_{\varphi}\frac{\left(
\varphi^{\left(  1\right)  },\operatorname*{T}_{h_{l}^{\operatorname*{ave}}%
}\varphi^{\left(  1\right)  }\right)  }{\left\Vert \varphi\right\Vert ^{2}%
}\text{,}%
\]
which corresponds to the last inequality of the lemma.
\end{proof}

We now return to completing the proof of Theorem~\ref{thm:TreeRecoKCSP}.
\ \noindent

\begin{proof}
[Proof of the lower bound in Theorem \textbf{\ref{thm:TreeRecoKCSP}}]If
$\theta=1$, $\mathsf{\operatorname*{T}}_{1}$ is the identity operator whence
$(\varphi^{\left(  1\right)  },\mathsf{\operatorname*{T}}_{1}\varphi^{\left(
1\right)  })=\operatorname*{I}_{1}\left(  \varphi\right)  $. We have therefore
$F_{k}(1)=1/\Omega_{k}$. Now, expanding in Fourier series we get,
\[
(\varphi^{\left(  1\right)  },\mathsf{\operatorname*{T}}_{\theta}%
\,\varphi^{\left(  1\right)  })=\sum_{Q\subseteq\lbrack k-1]}|\left(
\varphi^{\left(  1\right)  },\gamma_{Q}\right)  |^{2}\;\theta^{|Q|}%
=\sum_{Q\subseteq\lbrack k],Q\ni\left\{  i\right\}  }\left\vert \left(
\varphi,\gamma_{Q}\right)  \right\vert ^{2}\;\theta^{|Q|-1}\,\text{.}%
\]
By the \emph{Fourier expansion condition},
\begin{equation}
F_{k}(\theta)\leq e^{-Ck(1-\theta)}/\Omega_{k}\text{.} \label{f10}%
\end{equation}
Now fix $\alpha=(1-\delta)(\Omega_{k}\log k)/k$, whence, by Lemma
\ref{lemma:FirstStep}, $h_{1}^{\operatorname*{ave}}\leq1-k^{-1+\delta}$, and
$h_{1}$ is supported on non-negative reals. Using Eq. (\ref{eq:ineq3}), we get
$\widehat{h}_{2}^{\operatorname*{av}}\leq e^{-Ck^{\delta}}/\Omega_{k}$, and
therefore,
\[
h_{2}^{\operatorname*{av}}\leq1-\exp\{-2(1-\delta)e^{-Ck^{\delta}}\log
k\,\}\leq e^{-Ck^{\delta}/2}\,.
\]
On the other hand, from the Eq. (\ref{fourier3}), we obtain the following
bounds for $F_{k}(\theta)$, $R_{k}\left(  \theta\right)  $:
\[
F_{k}(\theta)\leq2\mathbb{E}_{\varphi}\left[  \frac{\sum_{i=1}^{k-1}|\left(
\varphi^{\left(  1\right)  },\gamma_{\left\{  i\right\}  }\right)  |^{2}%
}{\left\Vert \varphi\right\Vert ^{2}}\,\right]  \theta+2\mathbb{E}_{\varphi
}\left[  \frac{\operatorname*{I}_{1}\left(  \varphi\right)  }{\left\Vert
\varphi\right\Vert ^{2}}\,\right]  \theta^{2}\leq\left(  Ae^{-Ck/2}%
\theta+\theta^{2}\right)  /\Omega_{k}\text{.}%
\]

On the other hand,
\[
R_{k}(\theta)\leq2\mathbb{E}_{\varphi}{}_{\mathbf{i}}\left[  \frac
{2\operatorname*{I}_{1}\left(  \varphi\right)  }{\left\Vert \varphi\right\Vert
^{2}}\sum_{i=1}^{k-1}|\left(  \varphi^{\left(  1\right)  },\gamma_{\left\{
i\right\}  }\right)  |^{2}\right]  \,\theta^{2}+2\mathbb{E}_{\varphi}\left[
\frac{2\operatorname*{I}_{1}\left(  \varphi\right)  }{\left\Vert
\varphi\right\Vert ^{2}}\sum_{Q\subseteq\lbrack k-1]}|\left(  \varphi^{\left(
1\right)  },\gamma_{Q}\right)  |\,\right]  \theta^{3}\leq(Ae^{-Ck/2}\theta
^{2}+k^{a}\theta^{3})/\Omega_{k}\,,
\]
Therefore, for all $\ell$ we have
\[
h_{\ell+1}^{\operatorname*{av}}\leq1-e^{-k\alpha\lbrack F_{k}(h_{\ell
}^{\operatorname*{av}})+R_{k}(h_{\ell}^{\operatorname*{av}})]}\leq
(1-\delta)\log k(2Ae^{-Ck/2}h_{\ell}^{\operatorname*{av}}+2k^{a}(h_{\ell
}^{\operatorname*{av}})^{3/2})\,.
\]
which implies $h_{\ell}^{\operatorname*{av}}\rightarrow0$ if, for some
$\ell>0$, $h_{\ell}^{\operatorname*{av}}\leq k^{-5a}$, thus finishing the proof.
\end{proof}


\section{Reconstruction on Trees to Graphs: the case of proper $q$ colorings}

\label{sec:ColorRec}

In this section we prove that the set of solutions of the proper $q$-coloring
ensemble satisfies the \emph{sphericity} condition described in the section
\ref{sec:generalstrategy}.

Given two assignments $\underline{x}^{\left(  1\right)  }$, $\underline
{x}^{\left(  2\right)  }$ of the variables $x_{1},\ldots,x_{n}$, their joint
type $v_{\underline{x}^{\left(  1\right)  },\underline{x}^{\left(  2\right)
}}$ is the $q\times q$ matrix with $v_{\underline{x}^{\left(  1\right)
},\underline{x}^{\left(  2\right)  }}\left(  i,j\right)  \overset{def}{=}%
\frac{1}{n}\#\left\{  t\in G:\underline{x}^{\left(  1\right)  }\left(
t\right)  =i\text{ and }\underline{x}^{\left(  2\right)  }\left(  t\right)
=j\right\}  $. We consider random assignments $\underline{\mathbf{x}}^{\left(
1\right)  }$, $\underline{\mathbf{x}}^{\left(  2\right)  }$ taken uniformly
and independently over all the satisfying assignments of a random instance of
the $q$-coloring model with edge-variable density $\alpha$. Our purpose is to
prove that for all $\delta>0$, $||v_{\underline{\mathbf{x}}^{\left(  1\right)
},\underline{\mathbf{x}}^{\left(  2\right)  }}-\overline{\nu}%
||_{\operatorname*{TV}}\leq\delta$ w.h.p., where $\overline{v}$ is the matrix
with all entries equal to $1/q^{2} $.

Our argument makes crucial use of the following estimate for the partition
function from \cite{AchlioCojas}.

\begin{lemma}
[{\cite[Lemma 7]{AchlioCojas}}]\label{acojas}Let $Z$ be the number of
satisfying assignments of a random instance of the $q$-coloring model with
edge-variable density $\alpha<q\log q$, then
\[
\mathbf{E}Z\geq\Omega\left(  \frac{1}{n^{(q-1)/2}}\right)  \left[  q\left(
1-\frac{1}{q}\right)  ^{\alpha}\right]  ^{n}\,\text{,}%
\]
and, for some function $f(n)$ of order $o(n)$, we have $\operatorname*{Prob}%
\left(  Z<e^{-f(n)}\mathbf{E}\left[  Z\right]  \right)  \rightarrow0$ as
$n\rightarrow\infty$.
\end{lemma}

Let us introduce some notation. If $w$ is a vector of lenght $q$ and $v$ is a
$q\times q$ matrix $v$, let $\mathcal{H}$ and $\mathcal{E}$ denote their
entropy an their enrgy respectively, where
\begin{align*}
\mathcal{H}(v)  &  =-%
{\textstyle\sum\limits_{i,j}}
v\left(  i,j\right)  \log v\left(  i,j\right)  \,\text{, \quad}\mathcal{H}%
(w)=-%
{\textstyle\sum\limits_{i}}
w\left(  i\right)  \log w\left(  i\right) \\
\mathcal{E}(v)  &  =\log\left(  1-%
{\textstyle\sum\limits_{i}}
\left(
{\textstyle\sum\limits_{j}}
v\left(  i,j\right)  \right)  ^{2}-%
{\textstyle\sum\limits_{j}}
\left(
{\textstyle\sum\limits_{i}}
v\left(  i,j\right)  \right)  ^{2}+%
{\textstyle\sum\limits_{i,j}}
v\left(  i,j\right)  ^{2}\right)  \,\text{,\quad}\mathcal{E}(w)=\log\left(  1-%
{\textstyle\sum\limits_{i}}
w\left(  i\right)  ^{2}\right)
\end{align*}
Let $\mathcal{B}_{q}^{\epsilon}$ consists of all the $q$-vectors $w$ with
nonegative entries such that $%
{\textstyle\sum\limits_{i}}
w\left(  i\right)  =1$ and $\left\Vert w-\overline{w}\right\Vert ^{2}%
>\epsilon$. Similarly, let $\mathcal{B}_{q\times q}^{\delta,\epsilon}$ be the
set of all the $q\times q $ matrices with nonegative entries such that
$\left\Vert \left(  v-\overline{v}\right)  1\right\Vert ^{2}\leq\delta$,
$\left\Vert 1^{t}\left(  v-\overline{v}\right)  \right\Vert ^{2}\leq\delta$
and $\left\Vert v-\overline{v}\right\Vert ^{2}\geq\epsilon$.

Our goal in this section is to prove the following theorem.

\begin{theorem}
\label{uniform}Let $\underline{\mathbf{x}}^{\left(  1\right)  }$,
$\underline{\mathbf{x}}^{\left(  2\right)  }$ be random assignments taken
uniformly and independently over all the satisfying assignments of a random
instance of the $q$-coloring model with edge-variable density $\alpha
$\textbf{.} If $\alpha<\left(  q-1\right)  \log\left(  q-1\right)  $, then for
any $\epsilon>0 $,%
\[
\operatorname*{Prob}\left(  \left\Vert v_{\underline{\mathbf{x}}^{\left(
1\right)  },\underline{\mathbf{x}}^{\left(  2\right)  }}-\overline
{v}\right\Vert ^{2}>\epsilon\right)  \text{ }\rightarrow0\text{ as
}n\rightarrow\infty\text{.}%
\]

\end{theorem}

We will present several lemmas before returning to the proof of the Theorem.
First we introduce estimations concerning an additive functional depending on
the energy and entropy of a vector of lenght $q$.

\begin{lemma}
\label{boundy}If $w\in\mathcal{B}_{q}^{\epsilon}$, then $\mathcal{H}%
(w)+\alpha\mathcal{E}(w)\leq\left[  \mathcal{H}(\overline{w})+\alpha
\mathcal{E}(\overline{w})\right]  -\frac{\alpha\epsilon}{2\left(
1-1/q\right)  }\,.$
\end{lemma}

\begin{proof}
Notice that $\left[  \mathcal{H}(\overline{w})+\alpha\mathcal{E}(\overline
{w})\right]  -\left[  \mathcal{H}(w)+\alpha\mathcal{E}(w)\right]  =\alpha
\log\left(  \frac{1-\left\Vert \overline{w}\right\Vert ^{2}}{1-\left\Vert
w\right\Vert ^{2}}\right)  $. This quantity is bounded below by $\alpha
\log\left(  1+\frac{\epsilon}{1-1/q}\right)  $, and therefore by $\frac
{\alpha\epsilon}{2\left(  1-1/q\right)  }$.
\end{proof}

\begin{lemma}
\label{balanced}Let $\underline{\mathbf{x}}$ be a random assignment of the
variables taken uniformly over all the satisfying assignments of a random
instance of the $q$-coloring model with edge-variable density $\alpha<q\log
q$. Then, for any $\epsilon>0$,%
\[
\operatorname*{Prob}\left(  \left\Vert w_{\underline{\mathbf{x}}}-\overline
{w}\right\Vert ^{2}>\epsilon\right)  \rightarrow0\text{ as }n\rightarrow\infty
\]
where $w$ is the vector with $q$ entries such that $w_{\underline{\mathbf{x}}%
}\left(  i\right)  =\frac{1}{n}\#\left\{  v\in G:\underline{\mathbf{x}}%
_{v}=i\right\}  $ and $\overline{w}$ is the vector with all entries equal to
$1/q$.
\end{lemma}

\begin{proof}
Given a property $P$, denote by $Z(P)$, the number of satisfying assignments
for shich $P$ holds. Choose $\xi$ such that $\xi<\frac{\alpha\epsilon
}{2\left(  1-1/q\right)  }$. We have that
\[
\operatorname*{Prob}\left(  \left\Vert w_{\underline{\mathbf{x}}}-\overline
{w}\right\Vert ^{2}>\epsilon\right)  =\mathbf{E}\left[  Z\left(  \left\Vert
w_{\underline{\mathbf{x}}}\right\Vert ^{2}>\epsilon+1/q\right)  /Z\right]
\text{,}%
\]
an expression that we can bound by
\[
\frac{\mathbf{E}\left[  Z\left(  \left\Vert w_{\underline{\mathbf{x}}%
}\right\Vert ^{2}>\epsilon+1/q\right)  \right]  }{e^{-n\xi}\mathbf{E}\left[
Z\right]  }+\operatorname*{Prob}\left(  Z<e^{-n\xi}\mathbf{E}\left[  Z\right]
\right)  \text{.}%
\]
Now, according to the Lemma \ref{acojas}, $\operatorname*{Prob}\left(
Z<e^{-n\xi}\mathbf{E}\left[  Z\right]  \right)  \rightarrow0$, and therefore
it is enough to show that the term $\mathbf{E}\left[  Z\left(  \left\Vert
w_{\underline{\mathbf{x}}}\right\Vert ^{2}>\epsilon+1/q\right)  \right]
/e^{-n\xi}\mathbf{E}\left[  Z\right]  $ vanishes.

Denote by $\mathcal{G}_{\epsilon}$ the set of all vectors $\ell$, with
nonegative integer entries, such that $%
{\textstyle\sum\limits_{i=1}^{q}}
\left(  \ell_{i}/n\right)  =1$ and\linebreak\ $%
{\textstyle\sum\limits_{i=1}^{q}}
\left(  \ell_{i}/n\right)  ^{2}>\epsilon+1/q$, and denote by $\Omega_{w}$ the
set of assignments $\underline{x}$ such that $w_{\underline{x}}$ is equal to
the vector $w$. Now,
\begin{align}
\mathbf{E}\left[  Z\left(  \left\Vert w_{\underline{\mathbf{x}}}\right\Vert
^{2}>\epsilon+1/q\right)  \right]   &  =%
{\textstyle\sum\limits_{\ell\in\mathcal{G}_{\epsilon}}}
{\textstyle\sum_{\underline{x}\in\Omega_{\ell/n}}}
\operatorname*{Prob}\left(  \underline{x}\text{ is a satisfying assignment}%
\right) \label{ref1}\\
&  =%
{\textstyle\sum\limits_{\ell\in\mathcal{G}_{\epsilon}}}
\frac{n!}{%
{\textstyle\prod\limits_{i=1}^{q}}
\ell_{i}!}\left(  \left[  \frac{n}{n-1}\right]  \left[  1-%
{\textstyle\sum\limits_{i=1}^{q}}
\left(  \ell_{i}/n\right)  ^{2}\right]  \right)  ^{\alpha n}\nonumber\\
&  \leq%
{\textstyle\sum\limits_{\ell\in\mathcal{G}_{\epsilon}}}
3q^{2q}\sqrt{n}\exp\left(  n\left[  \mathcal{H}\left(  \ell/n\right)
+c_{n}\mathcal{E}\left(  \ell/n\right)  \right]  \right) \nonumber\\
&  \leq3q^{2q}\sqrt{n}\left\vert \mathcal{G}_{\epsilon}\right\vert \sup
_{\ell\in\mathcal{G}_{\epsilon}}\left\{  \exp\left(  n\left[  \mathcal{H}%
\left(  \ell/n\right)  +c_{n}\mathcal{E}\left(  \ell/n\right)  \right]
\right)  \right\}  \,\text{.}\nonumber
\end{align}
Here $\left\vert \mathcal{G}_{\epsilon}\right\vert $ is the number of elements
of $\mathcal{G}_{\epsilon}$, which is bounded by $n^{q}$. Notice also that if
$\ell\in\mathcal{G}_{\epsilon}$, then $\ell/n\in\mathcal{B}_{q}^{\epsilon}$,
so that by Lemma~\ref{boundy},
\begin{align}
\mathcal{H}\left(  \ell/n\right)  +\alpha\mathcal{E}\left(  \ell/n\right)   &
\leq\left[  \mathcal{H}(j_{q})+\alpha\mathcal{E}(j_{q})\right]  -\frac
{\alpha\epsilon}{2\left(  1-1/q\right)  }\label{ref2}\\
&  =\log q+\alpha\log\left(  1-1/q\right)  -\frac{\alpha\epsilon}{2\left(
1-1/q\right)  }\text{.}\nonumber
\end{align}
On the other hand by the Lemma \ref{acojas}, there is some constant $C$ such
that%
\begin{equation}
e^{-n\xi}\mathbf{E}\left[  Z\right]  \geq\frac{C}{n^{(q-1)/2}}e^{-n\xi}\left[
q\left(  1-\frac{1}{q}\right)  ^{\alpha}\right]  ^{n}\text{.} \label{ref3}%
\end{equation}
Combining Eq. (\ref{ref1}), (\ref{ref2}) and (\ref{ref3}), we have that for a
polynomial $p\left(  n\right)  $ of degree $3q/2$,
\begin{equation}
\frac{\mathbf{E}\left[  Z\left(  \left\Vert w_{\underline{\mathbf{x}}%
}-\overline{w}\right\Vert ^{2}>\epsilon\right)  \right]  }{e^{-n\xi}%
\mathbf{E}\left[  Z\right]  }\leq p(n)\exp\left(  n\left[  \xi-\frac
{\alpha\epsilon}{2\left(  1-1/q\right)  }\right]  \right)  \,\text{.}
\label{qq}%
\end{equation}
From (\ref{qq}), it is now clear that $\frac{\mathbf{E}\left[  Z\left(
\left\Vert w_{\underline{\mathbf{x}}}-\overline{w}\right\Vert ^{2}%
>\epsilon\right)  \right]  }{e^{-n\xi}\mathbf{E}\left[  Z\right]  }%
\rightarrow0$ as $n\rightarrow\infty$, due to the fact that $\xi-\frac
{\alpha\epsilon}{2\left(  1-1/q\right)  }<0$.
\end{proof}

Next, our objective is to work with the quantity $\kappa_{q}^{\delta,\epsilon
}$, which we define as the upper limit of the interval (indeed, easy to see
that this is an interval) consisting of the values $c$ such that
\[
\sup_{v\in\mathcal{B}_{q\times q}^{\delta,\epsilon}}\mathcal{H}%
(v)+c\mathcal{E}(v)\leq\mathcal{H}(\overline{v})+\alpha\mathcal{E}%
(\overline{v})\text{.}%
\]
To motivate, let us recall that an important part of the second moment
argument of Achlioptas and Naor \cite[Theorem 7]{achlioassaf} (in showing that
the chromatic number $\chi\left[  G\left(  n,d/n\right)  \right]  $
concentrated on two possible values), relied on an optimization of the
expression $\mathcal{H}(v)+\alpha\mathcal{E}(v)$ over the Birkoff polytope
$\mathcal{B}_{q\times q}$ of the $q\times q$ doubly stochastic matrices. In
particular, they proved that, as long as $\alpha\leq(q-1)\log(q-1)$, one has
\begin{equation}
\sup_{v\in\mathcal{B}_{q\times q}}\mathcal{H}(v)+\alpha\mathcal{E}%
(v)=\mathcal{H}(\overline{v})+\alpha\mathcal{E}(\overline{v})\,. \label{eqq}%
\end{equation}

\noindent Since $\mathcal{B}_{q\times q}^{0,\epsilon}\subseteq\mathcal{B}%
_{q\times q}$, we have $\kappa_{q}^{0,\epsilon}\geq\left(  q-1\right)
\log\left(  q-1\right)  $. The next lemma says that $\underset{v\in
\mathcal{B}_{q\times q}^{\delta,\epsilon}}{\sup}\mathcal{H}(v)+\alpha
\mathcal{E}(v)$ is in fact `separated' from $\mathcal{H}(\overline{v}%
)+\alpha\mathcal{E}(\overline{v})$, provided that $\alpha<\kappa_{q}%
^{\delta,\epsilon}$.

\begin{lemma}
\label{lemita}Suppose that $v\in\mathcal{B}_{q\times q}^{\delta,\epsilon}$
where $\epsilon>2\delta$, then, if $\alpha<\kappa_{q}^{\delta,\epsilon} $, we
have that
\[
\left[  \mathcal{H}(v)+\alpha\mathcal{E}(v)\right]  \leq\left[  \mathcal{H}%
(\overline{v})+\alpha\mathcal{E}(\overline{v})\right]  -\frac{\left(
\kappa_{q}^{\delta,\epsilon}-\alpha\right)  }{2\left(  1-1/q\right)  ^{2}%
}\left[  \epsilon-2\delta\right]  \,.
\]

\end{lemma}

\begin{proof}
Indeed,
\begin{align*}
\left[  \mathcal{H}(\overline{v})+\alpha\mathcal{E}(\overline{v})\right]
-\left[  \mathcal{H}(v)+\alpha\mathcal{E}(v)\right]   &  =\left[
\mathcal{H}(\overline{v})+\kappa_{q}^{\delta,\epsilon}\mathcal{E}(\overline
{v})\right]  -\left[  \mathcal{H}(v)+\kappa_{q}^{\delta,\epsilon}%
\mathcal{E}(v)\right]  +\left(  \kappa_{q}^{\delta,\epsilon}-\alpha\right)
\left[  \mathcal{E}(v)-\mathcal{E}(\overline{v})\right] \\
&  \geq\left(  \kappa_{q}^{\delta,\epsilon}-\alpha\right)  \left[  \log\left(
1+\frac{1}{\left(  1-1/q\right)  ^{2}}\left[  \left\Vert v-\overline
{v}\right\Vert ^{2}-\left\Vert \left(  v-\overline{v}\right)  1\right\Vert
^{2}-\left\Vert 1^{t}\left(  v-\overline{v}\right)  \right\Vert ^{2}\right]
\right)  \right] \\
&  \geq\frac{\left(  \kappa_{q}^{\delta,\epsilon}-\alpha\right)  }{2\left(
1-1/q\right)  ^{2}}\left[  \epsilon-2\delta\right]  \text{.}%
\end{align*}

\end{proof}

\begin{lemma}
\label{local}Given $\epsilon>0$ and $\alpha<\alpha_{q}=\left(  q-1\right)
\log\left(  q-1\right)  $, there exists $\delta>0$ such that $\kappa
_{q}^{\delta,\epsilon}\geq\alpha$.
\end{lemma}

\begin{proof}
Assume the contrary, then there exists a sequence $\delta_{n}\downarrow0$ such
that $\kappa_{q}^{\delta_{n},\epsilon}<\alpha$ for each $n$. Due to the
continuity of $\exp(\mathcal{H}(v)+\alpha\mathcal{E}(v))$ in the compact set
$\mathcal{B}_{q\times q}^{\delta,\epsilon}$, the supremum of $\exp
(\mathcal{H}(v)+\alpha_{q}\mathcal{E}(v))$ is reached at a matrix
$v_{\delta_{n}}\in\mathcal{B}_{q\times q}^{\delta_{n},\epsilon}\subseteq
\mathcal{P}_{q\times q}$, and due to the compactness of $\mathcal{P}_{q\times
q}$, a subsequence $\left\{  v_{\delta_{n_{k}}}\right\}  _{k\geq1}$ of these
matrices converges in $\mathcal{P}_{q\times q}$ to a matrix $v\in$
$\mathcal{B}_{q\times q}^{0,\epsilon}$. Therefore $\mathcal{H}(v)+\alpha
\mathcal{E}(v)\leq\mathcal{H}(\overline{v})+\alpha\mathcal{E}(\overline
{v})-\frac{\left(  \alpha_{q}-\alpha\right)  \epsilon}{2\left(  1-1/q\right)
^{2}}$. On the other hand,
\[
\mathcal{H}(v)+\alpha\mathcal{E}(v))\geq\liminf_{k\rightarrow\infty
}\mathcal{H}(v_{\delta_{n_{k}}})+\alpha\mathcal{E}\left(  v_{\delta_{n_{k}}%
}\right)  \geq\mathcal{H}(\overline{v})+\alpha\mathcal{E}\left(  \overline
{v}\right)  \text{, }%
\]
obtaining a contradiction.
\end{proof}

\begin{proof}
[\textbf{Proof of Theorem \ref{uniform}}]Given a property $P$, denote by
$Z^{(2)}\left(  P\right)  $, the number of pairs of satisfying assignments for
which $P$ holds. Take $\alpha^{\prime}$ such that $\alpha<\alpha^{\prime
}<\left(  q-1\right)  \log\left(  q-1\right)  $ and use Lemma~\ref{local} to
choose $\delta$ such that $\kappa_{q}^{\delta,\epsilon}\geq\alpha^{\prime}$,
guaranteeing also that $2\delta<\epsilon$. Now, let $\xi$ be a positive real
such that $2\xi<\frac{\left(  \alpha^{\prime}-\alpha\right)  }{2\left(
1-1/q\right)  ^{2}}\left[  \epsilon-2\delta\right]  $. We have that
\[
\operatorname*{Prob}\left(  \left\Vert v_{\underline{\mathbf{x}}^{\left(
1\right)  },\underline{\mathbf{x}}^{\left(  2\right)  }}-\overline
{v}\right\Vert ^{2}>\epsilon\right)  =\mathbf{E}\left[  Z^{(2)}\left(
\left\Vert v_{\underline{\mathbf{x}}^{\left(  1\right)  },\underline
{\mathbf{x}}^{\left(  2\right)  }}-\overline{v}\right\Vert ^{2}>\epsilon
\right)  /Z^{2}\right]  \text{,}%
\]
which is bounded by the addition of the terms $E\left[  {Z^{(2)}}\left(
v_{\underline{\mathbf{x}}^{\left(  1\right)  },\underline{\mathbf{x}}^{\left(
2\right)  }}{\in\mathcal{B}_{q\times q}^{\delta,\epsilon}}\right)  \right]
/{e^{-2n\xi}\mathbf{E}\left[  Z\right]  ^{2}}$, \newline$\operatorname*{Prob}%
\left(  Z<e^{-n\xi}\mathbf{E}\left[  Z\right]  \right)  $,
$\operatorname*{Prob}\left(  \left\Vert \left(  v_{\underline{\mathbf{x}%
}^{\left(  1\right)  },\underline{\mathbf{x}}^{\left(  2\right)  }}%
-\overline{v}\right)  1\right\Vert ^{2}>\epsilon\right)  $ and
$\operatorname*{Prob}\left(  \left\Vert 1^{t}\left(  v_{\underline{\mathbf{x}%
}^{\left(  1\right)  },\underline{\mathbf{x}}^{\left(  2\right)  }}%
-\overline{v}\right)  \right\Vert ^{2}>\epsilon\right)  $. Now,
Lemma~\ref{acojas} implies that the second term vanishes and lemma
\ref{balanced} implies that the last two terms go to zero. Therefore, to show
that $\operatorname*{Prob}\left(  \left\Vert v_{\underline{\mathbf{x}%
}^{\left(  1\right)  },\underline{\mathbf{x}}^{\left(  2\right)  }}%
-\overline{v}\right\Vert ^{2}>\epsilon\right)  \rightarrow0$ is sufficient to
prove that the term $\mathbf{E}\left[  Z^{(2)}\left(  v_{\underline
{\mathbf{x}}^{\left(  1\right)  },\underline{\mathbf{x}}^{\left(  2\right)  }%
}\in\mathcal{B}_{q\times q}^{\delta,\epsilon}\right)  \right]  /e^{-2n\xi
}\mathbf{E}\left[  Z\right]  ^{2}$ vanishes.

Denoting by $\mathcal{G}_{\epsilon,\delta}$ the set of all $q\times q$
matrices $L$, with nonegative integer entries, such that $L/n\in
\mathcal{B}_{q\times q}^{\delta,\epsilon}$, and denoting by $\Omega_{v}$ the
set of pairs of colorings $x_{1},x_{2}$ such that $v_{x_{1},x_{2}}$ is equal
to the matrix $v$, we have
\begin{align*}
\mathbf{E}\left[  Z^{(2)}\left(  v_{\underline{\mathbf{x}}^{\left(  1\right)
},\underline{\mathbf{x}}^{\left(  2\right)  }}\in\mathcal{B}_{q\times
q}^{\delta,\epsilon}\right)  \right]   &  =%
{\textstyle\sum\limits_{L\in\mathcal{G}_{\epsilon,\delta}}}
{\textstyle\sum_{x_{1},x_{2}\in\Omega_{L/n}}}
\operatorname*{Prob}\left(  x_{1}\text{ and }x_{2}\text{ are satisfying
assignments}\right) \\
&  =%
{\textstyle\sum\limits_{L\in\mathcal{G}_{\epsilon}}}
\frac{n!}{%
{\textstyle\prod\limits_{i,j}}
L_{ij}!}\left[  \frac{n}{n-1}\right]  ^{\alpha n}\left(  1-%
{\textstyle\sum\limits_{i}}
\left(
{\textstyle\sum\limits_{j}}
L_{ij}/n\right)  ^{2}-%
{\textstyle\sum\limits_{j}}
\left(
{\textstyle\sum\limits_{i}}
L_{ij}/n\right)  ^{2}+%
{\textstyle\sum\limits_{i,j}}
\left(  L_{ij}/n\right)  ^{2}\right)  ^{\alpha n}\\
&  \leq%
{\textstyle\sum\limits_{L\in\mathcal{G}_{\epsilon,\delta}}}
3q^{2q}\sqrt{n}\exp\left(  n\left[  \mathcal{H}\left(  L/n\right)  +\alpha
E\left(  L/n\right)  \right]  \right)  \text{.}%
\end{align*}
And now, because $\kappa_{q}^{\delta,\epsilon}\geq\alpha^{\prime}>\alpha$ and
$L/n\in\mathcal{B}_{q\times q}^{\delta,\epsilon}$ where $2\delta<\epsilon$, we
can invoke Lemma \ref{lemita} to get that%
\[
\left[  \mathcal{H}(L/n)+\alpha\mathcal{E}(L/n)\right]  \leq\left[
\mathcal{H}(\overline{v})+\alpha\mathcal{E}(\overline{v})\right]
-\frac{\left(  \alpha^{\prime}-\alpha\right)  }{2\left(  1-1/q\right)  ^{2}%
}\left[  \epsilon-2\delta\right]  \,.
\]
Therefore,
\[
\mathbf{E}\left[  Z^{(2)}\left(  v_{\underline{\mathbf{x}}^{\left(  1\right)
},\underline{\mathbf{x}}^{\left(  2\right)  }}\in\mathcal{B}_{q\times
q}^{\delta,\epsilon}\right)  \right]  \leq3q^{2q}\sqrt{n}\left\vert
\mathcal{G}_{\epsilon,\delta}\right\vert \left[  q\left(  1-1/q\right)
^{\alpha}\right]  ^{2n}\exp\left(  -n\frac{\left(  \alpha^{\prime}%
-\alpha\right)  }{2\left(  1-1/q\right)  ^{2}}\left[  \epsilon-2\delta\right]
\right)  \,,
\]
where $\left\vert \mathcal{G}_{\epsilon,\delta}\right\vert $ is the number of
elements in $\mathcal{G}_{\epsilon,\delta}$, which is bounded by $n^{q^{2}}$.
On the other hand by Lemma \ref{acojas}, we have that for some constant $C$,%
\[
e^{-2n\xi}\mathbf{E}\left[  Z\right]  ^{2}\geq\frac{C}{n^{(q-1)}}e^{-2n\xi
}\left[  q\left(  1-\frac{1}{q}\right)  ^{\alpha}\right]  ^{2n}\,.
\]
Hence, for a polynomial $p\left(  n\right)  $ of degree $q^{2}+q-1$, we have
\[
\frac{\mathbf{E}\left[  Z^{(2)}\left(  v_{\underline{\mathbf{x}}^{\left(
1\right)  },\underline{\mathbf{x}}^{\left(  2\right)  }}\in\mathcal{B}%
_{q\times q}^{\delta,\epsilon}\right)  \right]  }{e^{-2n\xi}\mathbf{E}\left[
Z\right]  ^{2}}\leq p(n)\exp\left\{  n\left(  2\xi-\frac{\left(
\alpha^{\prime}-\alpha\right)  }{2\left(  1-1/q\right)  ^{2}}\left[
\epsilon-2\delta\right]  \right)  \right\}  \,.
\]
Due to the fact that $2\xi<\frac{\left(  \alpha^{\prime}-\alpha\right)
}{2\left(  1-1/q\right)  ^{2}}\left[  \epsilon-2\delta\right]  $, it is now
clear that $\frac{\mathbf{E}\left[  Z^{(2)}\left(  v_{\underline{\mathbf{x}%
}^{\left(  1\right)  },\underline{\mathbf{x}}^{\left(  2\right)  }}%
\in\mathcal{B}_{q\times q}^{\delta,\epsilon}\right)  \right]  }{e^{-2n\xi
}\mathbf{E}\left[  Z\right]  ^{2}}\rightarrow0$ as $n\rightarrow\infty
$.\bigskip
\end{proof}

\bigskip

\textit{Acknowledgments}. The last two authors are grateful to Eric Vigoda and
Linji Yang for many insightful discussions on reconstruction problems, and for
their role in the early development of this project. The authors also
gratefully acknowledge the support and the hospitality of BIRS (Canada) and
DIMACS (USA), which provided ideal environs for carrying out a significant
part of this research collaboration.


\appendix

\section{Constrained partition function for binary CSP's}

\label{sec:SecondMom}

\bigskip

In this section, we prove Proposition~\ref{propo:SATCSP}. Given a random
$\mathrm{\operatorname*{CSP}}(n,p,\alpha)$ ensemble $\left\{  \varphi
_{a}\right\}  _{a=1}^{\alpha n}$, consider the statistic $L_{n}\left(
\varphi\right)  =\frac{1}{\alpha n}\#\left\{  a:\varphi_{a}=\varphi\right\}
$, and denote by $\mathrm{\operatorname*{CSP}}(n,p,\alpha;\widetilde{p}_{n})$
the ensemble $\left\{  \varphi_{a}\right\}  _{a=1}^{\alpha n}$ conditioned on
$L_{n}=\widetilde{p}_{n}$.. Also, denote by $\overline
{\mathrm{\operatorname*{CSP}}}(n,p,\alpha)$ the ensemble $\left\{  \varphi
_{a}\right\}  _{a=1}^{\alpha n}$ conditioned on $\left\Vert L_{n}-p\right\Vert
_{TV}<1/n^{1/2-\gamma}$, where $\gamma$ is a fixed positive constant. Because
$\operatorname*{Prob}\left(  \left\Vert L_{n}-p\right\Vert _{TV}%
\geq1/n^{1/2-\gamma}\right)  $ goes to zero (by the central limit theorem),
the probability measures induced by $\mathrm{\operatorname*{CSP}}(n,p,\alpha)$
and $\overline{\mathrm{\operatorname*{CSP}}}(n,p,\alpha)$ become equivalent as
$n\rightarrow\infty$.

A binary configuration $\underline{x}$ is said to be balanced if $\left\vert
\underline{x}\cdot\underline{1}\right\vert \leq1$. We will use $Z$ and $Z_{b}%
$, to denote the variable that counts the number of satisfying assignments and
balanced satisfying assignments, respectively, of a random CSP ensemble. Given
two binary assignments $\underline{x}^{(1)},\underline{x}^{(2)}$, we define
their overlap as $Q_{12}\overset{def}{=}\underline{x}^{(1)}\cdot\underline
{x}^{(2)}/n=\sum_{i=1}^{n}x_{i}^{(1)}x_{i}^{(2)}/n$. In other words
$(1-Q_{12})/2$ is the normalized Hamming distance of $\underline{x}^{(1)}$ and
$\underline{x}^{(2)}$.

\bigskip

The upper bound in Proposition~\ref{propo:SATCSP} follows from a first moment
calculation. In fact, for a random $\overline{\mathrm{\operatorname*{CSP}}%
}(n,p,\alpha)$, we have
\begin{align*}
\operatorname*{Prob}\left(  Z=0\right)   &  \leq\mathbf{E}\left[  Z\right]
=\sum_{x\in\left\{  -1,1\right\}  ^{k}}\operatorname*{Prob}\left(  x\text{ is
a satisfying assignment}\right)  =\sum_{\frac{x\cdot1}{n}=\theta}%
{\textstyle\prod\limits_{\varphi}}
\left\Vert \varphi\right\Vert _{\theta}^{2L_{n}\left(  \varphi\right)  \alpha
n}\\
&  \leq\exp\left(  n\left\{  \log2+\alpha%
{\textstyle\sum\limits_{\varphi}}
p\left(  \varphi\right)  \log\left\Vert \varphi\right\Vert ^{2}%
+\operatorname*{O}\left(  1/n^{1/2-\gamma}\right)  \right\}  \right)  \text{,}%
\end{align*}
and the last quantity goes to zero whenever $\alpha>\left(  1+\epsilon\right)
\widehat{\Omega}_{k}\log2$.

To establish the corresponding lower bound, we use the second moment method,
but first we need two lemmas.

\bigskip

\begin{lemma}
\label{lemma:secondmoment}Given a random $\mathrm{\operatorname*{CSP}%
}(n,p,\alpha;\widetilde{p}_{n})$ ensemble, let $Z_{\mathrm{b}}(|Q_{12}%
|\geq\delta)$ be the number of balanced solution pairs $\underline{x}^{(1)}$,
$\underline{x}^{(2)}\in\{+1,-1\}^{n}$ with overlap larger than $\delta$.
Then,
\[
\frac{\mathbf{E}\,\left[  Z_{\mathrm{b}}(|Q_{12}|\geq\delta)\right]  }{\left[
\mathbf{E}Z_{\mathrm{b}}\right]  ^{2}}\leq n\,\exp\left\{  n\left[
\sup\limits_{\theta\geq\delta}\Phi\left(  \theta\right)  \right]  \right\}
\,\text{,}%
\]
where
\[
\Phi(\theta)\overset{def}{=}H(\theta)+\alpha\mathbb{E}_{\varphi\sim
\widetilde{p}_{n}}\log\left\{  \frac{(\varphi,\operatorname*{T}_{\theta
}\varphi)}{\left\Vert \varphi\right\Vert ^{4}}\right\}  \text{,}%
\]
and $H(\theta)\equiv-\frac{1+\theta}{2}\log(1+\theta)-\frac{1-\theta}{2}%
\log(1-\theta)$.
\end{lemma}

\begin{proof}
For simplicity take $n$ to be even. Let $\varphi$ be a boolean function, and
let $\pi:\left[  k\right]  \rightarrow\left[  n\right]  $ be a uniform random
assignation for the variables in $\varphi$. Now, given two \emph{balanced}
vectors $\underline{x}^{\left(  1\right)  }$, $\underline{x}^{\left(
2\right)  }\in\left\{  -1,1\right\}  ^{n}$, we have
\[
\mathbb{E}_{\pi}\left[  \varphi\left(  \underline{x}_{\pi_{1}}^{\left(
1\right)  },\ldots\underline{x}_{\pi_{k}}^{\left(  1\right)  }\right)
\varphi\left(  \underline{x}_{\pi_{1}}^{\left(  2\right)  },\ldots
,\underline{x}_{\pi_{k}}^{\left(  2\right)  }\right)  \right]  =\left(
\varphi,\operatorname*{T}\nolimits_{\theta}\varphi\right)  \text{,}%
\]
where $\theta=Q_{12}$. Therefore, for some constant $C>0$,
\begin{align*}
\mathbb{E}_{\pi_{a}}Z_{\mathrm{b}}(|Q_{12}|  &  \geq\delta)=%
{\displaystyle\sum\limits_{\theta\geq\delta}}
{\displaystyle\sum\limits_{Q_{12}=\theta}}
{\textstyle\prod\limits_{\varphi}}
\left\vert \left(  \varphi,\operatorname*{T}\nolimits_{\theta}\varphi\right)
\right\vert ^{L_{n}\left(  \varphi\right)  \alpha n}\\
&  <%
{\displaystyle\sum\limits_{\theta\geq\delta}}
\frac{C}{n^{3/2}}\exp\left(  n\left\{  \mathcal{H}\left(  \frac{1+\theta}%
{4},\frac{1+\theta}{4},\frac{1-\theta}{4},\frac{1-\theta}{4}\right)  +\alpha%
{\textstyle\sum\limits_{\varphi}}
L_{n}\left(  \varphi\right)  \log\left(  \varphi,\operatorname*{T}%
\nolimits_{\theta}\varphi\right)  \right\}  \right)  \text{.}%
\end{align*}
where $\mathcal{H}\left(  \cdot\right)  $ is the entropy function. On the
other hand, for some positive $C^{\prime}$,
\begin{align*}
\mathbb{E}_{\pi_{a}}\,Z_{\mathrm{b}}  &  =%
{\displaystyle\sum\limits_{\underline{x}\text{ balanced}}}
{\textstyle\prod\limits_{\varphi}}
\left\Vert \varphi\right\Vert ^{2L_{n}\left(  \varphi\right)  \alpha n}\\
&  >\frac{C^{\prime}}{n^{1/2}}\exp\left(  n\left\{  \mathcal{H}\left(
\frac{1}{2},\frac{1}{2}\right)  +\alpha%
{\textstyle\sum\limits_{\varphi}}
L_{n}\left(  \varphi\right)  \log\left\Vert \varphi\right\Vert ^{2}\right\}
\right)  \text{.}%
\end{align*}
It is straightforward now to check that
\begin{equation}
\frac{\mathbf{E}Z_{\mathrm{b}}(|Q_{12}|\geq\delta)}{\mathbf{E}\left(
Z_{\mathrm{b}}\right)  ^{2}}<%
{\displaystyle\sum\limits_{\theta\geq\delta}}
\frac{C^{\prime\prime}}{n^{1/2}}\exp\left(  n\left[  \Phi\left(
\theta\right)  \right]  \right)  \label{eq:calc1}%
\end{equation}
and therefore $\frac{\mathbf{E}\,Z_{\mathrm{b}}(|Q_{12}|\geq\delta)}{\left(
\mathbf{E}Z_{\mathrm{b}}\right)  ^{2}}<n\exp\left(  n\left[  \sup
\limits_{\theta\geq\delta}\Phi\left(  \theta\right)  \right]  \right)  $.
\end{proof}

\bigskip

\begin{lemma}
Given a random $\mathrm{\operatorname*{CSP}}(n,p,\alpha;\widetilde{p})$
ensemble, if $\alpha\leq(1-\varepsilon)\Omega_{k,\widetilde{p}_{n}}\log2$,
where $\frac{1}{\Omega_{k,\widetilde{p}_{n}}}\overset{def}{=}\mathbb{E}%
_{\varphi\sim\widetilde{p}_{n}}\frac{2\operatorname*{I}_{1}\left(
\varphi\right)  }{\left\Vert \varphi\right\Vert ^{2}}$, then for any
$\delta>0$ there exists $C(\delta,\varepsilon)>0$ such that
\[
\mathbf{E}\,\left[  Z_{\mathrm{b}}(|Q_{12}|\geq\delta)\right]  \leq
e^{-n\left[  C(\delta,\epsilon)\right]  }\left(  \mathbf{E}Z_{\mathrm{b}%
}\right)  ^{2}\text{.}%
\]
Moreover, as $\delta\rightarrow0$, $C(\delta,\epsilon)=\Omega\left(
\delta^{2}\right)  $.
\end{lemma}

\begin{proof}
In view of the previous lemma, it is sufficient to prove that the function
$\theta\mapsto\Phi(\theta)$ achieves its maximum over the interval $[0,1]$
uniquely at $\theta=0$. To establish the second statement, then it will be
enough to prove that $-\Phi\left(  \theta\right)  =\Omega\left(  \theta
^{2}\right)  $ as $\theta\rightarrow0$.

Fix $\alpha\leq(1-\varepsilon)\Omega_{k}\log2\leq(1-\varepsilon)\widehat
{\Omega}_{k}\log2$. We will prove the thesis claim by considering three
different regimes for $\theta$: $0<\theta\leq e^{-ak}$, $e^{-ak}\leq\theta
\leq1-\varepsilon^{1/2}$ and $1-\varepsilon^{1/2}\leq\theta\leq1$, where $a$
is a small constant. In the first two intervals we will prove that the
derivative of $\Phi(\theta)$ with respect to $\theta$ is strictly negative.
Recalling that $\left\Vert \varphi\right\Vert ^{2}\geq1/2$, we have
\begin{align*}
\frac{\mathrm{d}\Phi}{\mathrm{d}\theta}  &  \leq-\mathrm{atanh}\,\theta
+k\alpha\mathbb{E}_{\varphi}\frac{(\varphi^{(1)},\operatorname*{T}_{\theta
}\varphi^{(1)})}{\left\Vert \varphi\right\Vert ^{4}}\\
&  \leq-\theta+2k\alpha\mathbb{E}_{\varphi}\frac{\sum_{i=1}^{k-1}%
|\varphi_{\{i\}}^{(1)}|^{2}}{\left\Vert \varphi\right\Vert ^{2}}%
\;\theta+2k\alpha\mathbb{E}_{\varphi}\frac{||\varphi^{(1)}||^{2}}{\left\Vert
\varphi\right\Vert ^{2}}\;\;\theta^{3}\\
&  \leq-\theta+Ae^{-Ck}\frac{\alpha}{\Omega_{k}}\theta+2k\frac{\alpha}%
{\Omega_{k}}\theta^{2}\leq-\frac{1}{2}\theta+4k\theta^{2}\,\text{,}%
\end{align*}
where we used (from Eq.~(\ref{fourier})) the hypothesis on low weight Fourier
coefficients. The last expression is strictly negative if $0<\theta<e^{-ak}$
for any $a>0$ and all $k$ large enough. The previous formula also shows
$-\Phi\left(  \theta\right)  =\Omega\left(  \theta^{2}\right)  $ as
$\theta\rightarrow0$.

Next assume $e^{-ak}\leq\theta\leq1-\varepsilon$. Using the hypothesis
$(\varphi^{\left(  1\right)  },\operatorname*{T}_{\theta}\varphi^{\left(
1\right)  })\leq e^{-Ck(1-\theta)}||\varphi^{\left(  1\right)  }||^{2}$, we
have
\begin{align*}
\frac{\mathrm{d}\Phi}{\mathrm{d}\theta}  &  \leq-\mathrm{atanh}\,\theta
+4k\alpha\mathbb{E}_{\varphi}\frac{||\varphi^{\left(  1\right)  }||^{2}%
}{||\varphi||^{4}}\;e^{-Ck\epsilon}\\
&  \leq-\mathrm{atanh}\,\theta+2k\frac{\alpha}{\Omega_{k}}e^{-Ck\sqrt
{\epsilon}}\leq-\mathrm{atanh}\theta+2\left(  \log2\right)  ke^{-Ck\epsilon
}\,,
\end{align*}
which is strictly negative if $\theta>e^{-ak}$ with, say, $a=(C\epsilon
^{2})/2$. Finally, we notice that, for $1-\varepsilon^{2}\leq\theta\leq1$, any
$\varepsilon$ small enough we have $H(\theta)\leq-\log2+\varepsilon/10$.
Further, using the fact that $(\varphi,\operatorname*{T}_{\theta}%
\varphi)=||\operatorname*{T}_{\theta^{1/2}}\varphi||^{2}$ is non-decreasing in
$\theta$
\[
\Phi(\theta)\leq-\log2+\frac{\varepsilon}{10}-\alpha\mathbb{E}_{\varphi}%
\log||\varphi||^{2}=-\log2+\frac{\varepsilon}{10}+\frac{\alpha}{\widehat
{\Omega}_{k}}\leq-\varepsilon\frac{\log2}{2}\,,
\]
which finishes the proof.
\end{proof}

\bigskip

\noindent\textbf{Conclusion of Proof of Proposition~\ref{propo:SATCSP}}. From
the previous lemma we have that for any fixed $\delta>0$,
\[
\frac{\mathbf{E}\,Z_{\mathrm{b}}^{2}}{\left(  \mathbf{E}\,Z_{\mathrm{b}%
}\right)  ^{2}}\leq\frac{\mathbf{E}\,\left[  Z_{\mathrm{b}}(|Q_{12}|<\left(
\delta/n\right)  ^{1/2})\right]  }{\left(  \mathbf{E}\,Z_{\mathrm{b}}\right)
^{2}}+e^{-\Omega\left(  \delta\right)  }\text{,}%
\]
while a calculation analogous to that in Eq. (\ref{eq:calc1}) and the fact
that $-\Phi\left(  \theta\right)  =\Omega\left(  \theta^{2}\right)  $, implies
that%
\[
\frac{\mathbf{E}\,\left[  Z_{\mathrm{b}}(|Q_{12}|<\left(  \delta/n\right)
^{1/2})\right]  }{\left(  \mathbf{E}\,Z_{\mathrm{b}}\right)  ^{2}}\leq%
{\displaystyle\sum\limits_{\theta<\left(  \delta/n\right)  ^{1/2}}}
\frac{C^{\prime\prime}}{n^{1/2}}\exp\left(  -n\Omega\left(  \theta^{2}\right)
\right)  \leq%
{\textstyle\int\limits_{-\delta}^{\delta}}
\exp\left(  -\Omega\left(  x^{2}\right)  \right)  dx\text{.}%
\]
Now, letting $\delta\rightarrow0$, it is clear that $\frac{\mathbf{E}%
\,Z_{\mathrm{b}}^{2}}{\left(  \mathbf{E}\,Z_{\mathrm{b}}\right)  ^{2}}$ tends
to $1$. This proves, by means of the Paley-Zygmund inequality, that for
$\alpha<(1-\varepsilon)\left(  \liminf\limits_{n\rightarrow\infty}%
\Omega_{k,\widetilde{p}_{n}}\right)  \log2$, a $\mathrm{\operatorname*{CSP}%
}(n,p,\alpha;\widetilde{p}_{n})$ ensemble is satisfiable w.h.p. The result
extends straightforwardly for a random\ $\mathrm{\operatorname*{CSP}%
}(n,p,\alpha)$, after noticing that $\Omega_{k,L_{n}}>\left(  1-\epsilon
\right)  \Omega_{k,p}$ with high probability. \hfill $\Box$

\bigskip
\end{document}